\pdfoutput=1
\documentclass[a4paper,UKenglish,thm-restate,pdfa]{lipics-v2021-modified}

\usepackage{mathtools}

\def\P{\ensuremath{\mathrm{P}}}
\def\NP{\ensuremath{\mathrm{NP}}}
\def\DP{\ensuremath{\mathrm{DP}}}
\def\coDP{\ensuremath{\mathrm{coDP}}}
\def\BH{\ensuremath{\mathrm{BH}}}
\def\coBH{\ensuremath{\mathrm{coBH}}}
\def\E{\ensuremath{\mathrm{E}}}
\def\EE{\ensuremath{\mathrm{EE}}}
\def\NE{\ensuremath{\mathrm{NE}}}
\def\NEE{\ensuremath{\mathrm{NEE}}}
\def\FP{\ensuremath{\mathrm{FP}}}
\def\UP{\ensuremath{\mathrm{UP}}}
\def\DisjNP{\ensuremath{\mathrm{DisjNP}}}

\def\coNP{\ensuremath{\mathrm{coNP}}}
\def\coNE{\ensuremath{\mathrm{coNE}}}
\def\coNEE{\ensuremath{\mathrm{coNEE}}}

\def\TFNP{\ensuremath{\mathrm{TFNP}}}

\def\TAUT{\ensuremath{\mathrm{TAUT}}}
\def\SAT{\ensuremath{\mathrm{SAT}}}

\def\TFNP{\ensuremath{\mathrm{TFNP}}}
\def\BPP{\ensuremath{\mathrm{BPP}}}
\def\PSPACE{\ensuremath{\mathrm{PSPACE}}}

\def\NPC{\ensuremath{\mathrm{C}}}
\def\DPC{\ensuremath{\mathrm{D}}}

\def\N{\ensuremath{\mathrm{\mathbb{N}}}}

\newcommand{\codecoDP}[4]{\langle 0^{#1}, 0^{#2}, #3, 0^{#4} \rangle}
\renewcommand{\c}[3]{\langle 0^{#1}, 0^{4(|#2|^{#1}+#1+|#3|)}, #2, #3 \rangle}

\renewcommand{\iff}{\,\mathop{\Longleftrightarrow}\,}

\DeclareMathOperator{\dom}{dom}
\DeclareMathOperator{\img}{img}
\DeclareMathOperator{\supp}{supp}

\def\leqmp{\ensuremath{\leq_\mathrm{m}^\mathrm{p}}}
\def\psim{\ensuremath{\leq^\mathrm{p}}}
\def\sqsubsetneq{\mathrel{\sqsubseteq\kern-0.92em\raise-0.15em\hbox{\rotatebox{313}{\scalebox{1.1}[0.75]{\(\shortmid\)}}}\scalebox{0.3}[1]{\ }}}
\def\sqsupsetneq{\mathrel{\sqsupseteq\kern-0.92em\raise-0.15em\hbox{\rotatebox{313}{\scalebox{1.1}[0.75]{\(\shortmid\)}}}\scalebox{0.3}[1]{\ }}}

\newcommand{\card}[1]{|#1|_c}

\newcommand\x{0}

\newcommand{\ifblock}[2]{\if\x0{#1}\fi\if\x1{#2}\fi} 

\hideLIPIcs

\title{Upward Translation of Optimal and P-Optimal Proof Systems in the Boolean Hierarchy over NP}
\titlerunning{Upward Translation of (P-)Optimal Proof Systems in the Boolean Hierarchy over NP}

\author{Fabian Egidy}{Julius-Maximilians-Universität Würzburg, Germany}{fabian.egidy@uni-wuerzburg.de}{https://orcid.org/0000-0001-8370-9717}{supported by the German Academic Scholarship Foundation.}
\author{Christian Glaßer}{Julius-Maximilians-Universität Würzburg, Germany}{christian.glasser@uni-wuerzburg.de}{}{}
\author{Martin Herold}{Max-Planck-Institut für Informatik, Saarbrücken, Germany}{mherold@mpi-inf.mpg.de}{}{Funded by the Deutsche Forschungsgemeinschaft (DFG, German Research
Foundation) – 399223600.
}
\authorrunning{F. Egidy, C. Glaßer, and M. Herold}
\Copyright{Fabian Egidy, Christian Glaßer, and Martin Herold}
\ccsdesc{Theory of computation~Proof complexity}
\ccsdesc{Theory of computation~Oracles and decision trees}
\keywords{Computational Complexity, Boolean Hierarchy, Proof Complexity, Proof Systems, Oracle Construction}
\nolinenumbers

\begin{document}
\maketitle
\begin{abstract}
We study the existence of optimal and p-optimal proof systems for classes in the Boolean hierarchy over $\NP$. Our main results concern $\DP$, i.e., the second level of this hierarchy:
\begin{itemize}
\item If all sets in $\DP$ have p-optimal proof systems, then all sets in $\coDP$ have p-optimal proof systems.
\item The analogous implication for optimal proof systems fails relative to an oracle.
\end{itemize}
As a consequence, we clarify such implications for all classes $\mathcal{C}$ and $\mathcal{D}$ in the Boolean hierarchy over $\NP$: either we can prove the implication or show that it fails relative to an oracle.

Furthermore, we show that the sets $\SAT$ and $\TAUT$ have p-optimal proof systems, if and only if all sets in the Boolean hierarchy over $\NP$ have p-optimal proof systems which is a new characterization of a conjecture studied by Pudlák.

\end{abstract}

\section{Introduction}
This paper contributes to the study of proof systems initiated by Cook and Reckhow \cite{cr79}. A proof system for a set $L$ is a polynomial-time computable function $f$ whose range is $L$. Cook and Reckhow motivate the study of proof systems with the $\NP = \coNP$ question: they consider propositional proof systems (pps), i.e., proof systems for the set of propositional tautologies ($\TAUT$). They show that there exists a pps with polynomially bounded proofs if and only if $\NP = \coNP$. This approach to the $\NP = \coNP$ question is called the Cook-Reckhow program \cite{bus96}. To obtain $\NP \not = \coNP$ one can either show that optimal pps (i.e., pps with at most polynomially longer proofs than any other pps) do not exist or show that a specific pps is optimal and has a non-polynomial lower bound on the length of proofs. This connection led to the investigation of upper and lower bounds for different pps \cite{kra19} as well as the existence of optimal and p-optimal\footnote{A stronger notion of optimal. We write (p-)optimal when the statement holds using optimal as well as p-optimal.} proof systems for general sets.

The latter question was explicitly posed by Krajíček and Pudlák \cite{kp89} in the context of finite consistency. They revealed the following connection between both concepts: optimal pps exist if and only if there is a finitely axiomatized theory $S$ that proves for every finitely axiomatized theory $T$ the statement ``$T$ has no proof of contradiction of length $n$'' by a proof of polynomial length in $n$. If optimal pps exist, then a weak version of Hilbert's program is possible, i.e., proving the ``consistency up to some feasible length of proofs'' of all mathematical theories \cite{pud96}. We refer to Krajíček \cite{kra95} and Pudlák \cite{pud98} for details on the relationship between proof systems and bounded arithmetic. More recently, Pudlák \cite{pud17} draws new connections of (p-)optimal proof systems and statements about incompleteness in the finite domain. 

Furthermore, proof systems have shown to be tightly connected to promise classes, especially pps to the class of disjoint $\NP$-pairs, called $\DisjNP$. Initiated by Razborov \cite{raz94}, who showed that the existence of p-optimal pps implies the existence of complete sets in $\DisjNP$, many further connections were investigated. More generally, Köbler, Messner and Torán \cite{kmt03} show that the existence of p-optimal proof systems for sets of the polynomial-time hierarchy imply complete sets for promise classes like $\UP, \NP \cap \coNP$, and $\BPP$. Beyersdorff, Köbler, and Messner \cite{bkm09} and Pudlák \cite{pud17} connect proof systems to function classes by showing that p-optimal proof systems for $\SAT$ imply complete sets for $\TFNP$. Questions regarding non-deterministic function classes can be characterized by questions about proof systems \cite{bkm09}. Beyersdorff \cite{bey04, bey06, bey07, bey10}, Beyersdorff and Sadowski \cite{bs11} and Glaßer, Selman, and Zhang \cite{gsz07, gsz09} show further connections between pps and disjoint $\NP$-pairs.

The above connections to important questions of complexity theory, bounded arithmetic, and promise classes motivate the investigation of the question ``which sets do have optimal proof systems'' posed by Messner \cite{mes00}. Krajíček and Pudlák \cite{kp89} were the first to study sufficient conditions for pps by proving that $\NE = \coNE$ implies the existence of optimal pps and $\E = \NE$ implies the existence of p-optimal pps. Köbler, Messner, and Torán \cite{kmt03} improve this result to $\NEE = \coNEE$ for optimal pps and $\EE = \NEE$ for p-optimal pps. Sadowski \cite{sad02} shows different characterizations for the existence of optimal pps, e.g., the uniformly enumerability of the class of all easy subsets of $\TAUT$. In certain settings one can prove the existence of optimal proof systems for different classes: e.g., by allowing one bit of advice \cite{ck07}, considering randomized proof systems \cite{hi10, hir10}, or using a weak notion of simulation \cite{ps10}. 

Messner \cite{mes00} shows that all nonempty\footnote{By our definition, $\FP$-functions are total, thus the empty set has no proof system. For the rest of this paper, we omit the word ``nonempty'' when referring to proof systems for all sets of a class, since this is only a technicality.} sets in $\P$ but not all sets in $\E$ have p-optimal proof systems. Similarly, all sets in $\NP$ but not all sets in $\coNE$ have optimal proof systems. Therefore, when going from smaller to larger complexity classes, there has to be a tipping point such that all sets contained in classes below this point have (p-)optimal proof systems, but some set contained in all classes above this point has no (p-)optimal proof systems. Unfortunately, oracle constructions tell us that for many classes between $\P$ and $\E$ (resp., $\NP$ and $\coNE$) the following holds: with relativizable proofs one can neither prove nor refute that p-optimal (resp., optimal) proof systems exist (e.g.\ $\coNP$ \cite{bgs75, kp89} and $\PSPACE$ \cite{bgs75, dek76}). Thus, with the currently available means it is not possible to precisely locate this tipping point, but we can rule out certain regions for its location. For this, we investigate how the existence of \mbox{(p-)optimal} proof system for all sets of the class $\mathcal{C}$ ``translate upwards'' to all sets of a class $\mathcal{D}$ with $\mathcal{C} \subseteq \mathcal{D}$. This rules out tipping points between $\mathcal{C}$ and $\mathcal{D}$.

\subparagraph*{Our Contribution.}  
Motivated by Messner's general question, we study the existence of \mbox{(p-)optimal} proof systems for classes inside the Boolean hierarchy over $\NP$. We use the expression ``the class $\mathcal{C}$ has (p-)optimal proof systems'' for ``all sets of a class $\mathcal{C}$ have \mbox{(p-)optimal} proof systems''. We say that two classes $\mathcal{C}$ and $\mathcal{D}$ are equivalent with respect to (p-)optimal proof systems if $\mathcal{C}$ has (p-)optimal proof systems if and only if $\mathcal{D}$ has (p-)optimal proof systems.

For the classes of the Boolean hierarchy over $\NP$, denoted by $\BH$, we identify three equivalence classes for p-optimal proof systems and three other classes for optimal proof systems. We also show that the classes of the bounded query hierarchy over $\NP$ are all equivalent for p-optimal proof systems and we identify two equivalence classes for optimal proof systems. Moreover, we show that relativizable techniques cannot prove all identified equivalence classes to coincide. These results follow from our main theorems:
\begin{enumerate}[(i)]
\item If $\DP$ has p-optimal proof systems, then $\coDP$ has p-optimal proof systems.
\item There exists an oracle relative to which $\coNP$ has p-optimal proof systems and $\coDP$ does not have optimal proof systems.
\end{enumerate}
Using the result by Köbler, Messner, and Torán that (p-)optimality is closed under intersection~\cite{kmt03} and two oracles by Khaniki \cite{kha22}, we obtain the equivalence classes visualized in Figure \ref{fig:results}, which cannot be proved to coincide with relativizable proofs.

\begin{figure}[h]
\begin{centering}
\includegraphics[width=0.8\textwidth]{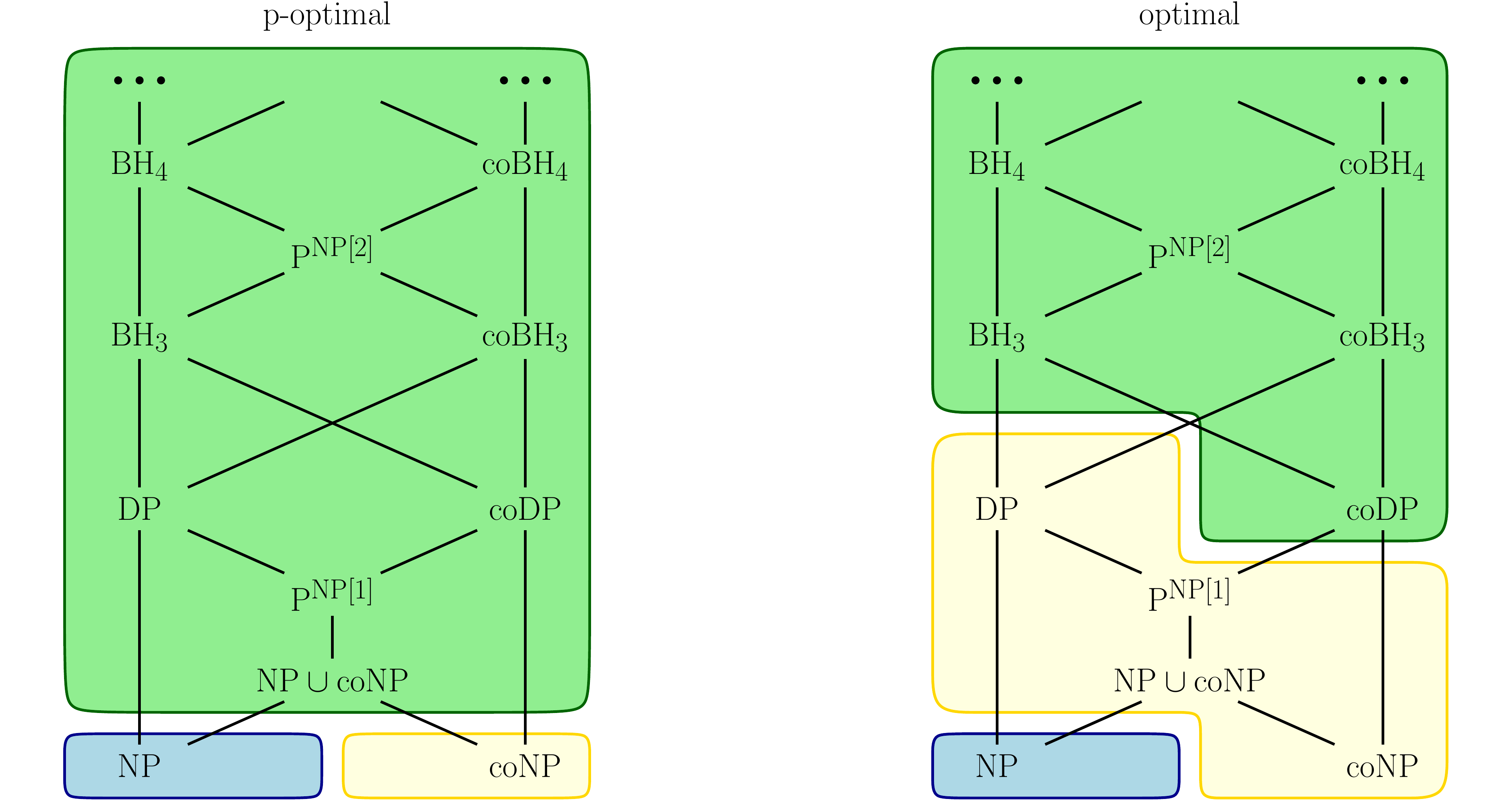}
  \caption{Equivalence classes for p-optimal proof systems (left) and optimal proof systems (right) in the Boolean hierarchy over $\NP$ and the bounded query hierarchy over $\NP$.}
  \label{fig:results}
\end{centering}
\end{figure}
This clarifies all questions regarding relativizably provable translations of (p-)optimal proof systems for classes in the Boolean hierarchy over $\NP$ and the bounded query hierarchy over $\NP$. We cannot expect to prove any further translations with the currently available means, because for every such translation there is an oracle against it. So we are dealing with an interesting situation: while p-optimal proof systems for $\DP$ {\em relativizably imply} p-optimal proof systems for $\coDP$, this does not hold for optimal proof systems. Similarly, all classes of the bounded query hierarchy over $\NP$ are equivalent with respect to p-optimal proof systems, but $\P^{\NP [1]}$ and $\P^{\NP [2]}$ cannot be shown to be equivalent with respect to optimal proof systems by a relativizable proof. The result drastically limits the potential locations of a tipping point in the $\BH$ and the bounded query hierarchy over $\NP$. They can only occur between two classes belonging to two different equivalence classes.

Furthermore, our results provide a new perspective on an hypothesis related to feasible versions of Gödel's incompleteness theorem: Pudlák \cite{pud17} studies several conjectures about incompleteness in the finite domain by investigating the (un)provability of sentences of some specific form in weak theories. These conjectures can also be expressed as the non-existence of complete sets in promise classes or non-existence of (p-)optimal proof systems for sets. Pudlák considers the conjecture $\mathsf{CON} \lor \mathsf{SAT}$ stating that $\TAUT$ does not have p-optimal proof systems or $\SAT$ does not have p-optimal proof systems. Khaniki \cite{kha22} proves this conjecture to be equivalent to $\mathsf{RFN}_1$, which is another conjecture considered by Pudlák. Our results show that both conjectures are equivalent to the statement that $\BH$ does not have p-optimal proof systems.

\section{Preliminaries}\label{sec:prelim}

Let $\Sigma \coloneqq \{0,1\}$ be the default alphabet and $\Sigma ^*$ be the set of finite words over $\Sigma$. We call subsets of $\Sigma^*$ languages and sets of languages classes. We denote the length of a word $w \in \Sigma^*$ by $|w|$\ifblock{ and the length of a set of words $A \subseteq \Sigma ^*$ by $\ell (A) \coloneqq \sum _{w \in A}|w|$. The empty word has length $0$ and is denoted by $\varepsilon$}{}. The $i$-th letter of a word $w$ for $0 \leq i < |w|$ is denoted as $w(i)$, i.e., $w=w(0)w(1)\cdots w(|w|-1)$. \ifblock{If $v$ is a (strict) prefix of $w$, we write $v \sqsubseteq w$ ($v \sqsubsetneq w$) or $w \sqsupseteq v$ ($w \sqsupsetneq v$).}{}

The set of all (positive) natural numbers is denoted as $\N$ ($\N^+$). We write the empty set as $\emptyset$. We identify $\Sigma^*$ with $\N$ through the polynomial time computable and invertible bijection $\Sigma^* \to \N; w \mapsto \sum _{i<|w|}(1+w(i))2^i$. This is a variant of the dyadic representation. Thus, we can treat words from $\Sigma^*$ as numbers from $\N$ and vice versa, which allows us to use notations, relations and operations of words for numbers and vice versa (e.g.\ we can define the length of a number by this).\ifblock{ Note that $x < y$ if and only if $x < _{\text{lex}} y$ for $x,y \in \Sigma ^*$, where $<_{\text{lex}}$ is the quasi-lexicographic ordering relation for words.}{} We resolve the ambiguity of $0^i$ and $1^i$ by always interpreting them as words from $\Sigma^*$. The cardinality of a set $A$ is denoted as $\card{A}$. For $\circ \in \{<, \leq, =, \geq, >\}$, a set $A \subseteq \Sigma^*$ and a number $n \in \N$ we define $A^{\circ n} \coloneqq \{w \in A \mid |w| \circ n\}$. For a clearer notation we use $\Sigma ^{\circ n}$ as ${\Sigma^*}^{\circ n}$ and $\Sigma^n$ for $\Sigma^{=n}$. The operators $\cup$, $\cap$, and $\setminus$ denote the union, intersection and set-difference. We denote the complement of a set $A$ relative to $\Sigma ^*$ as $\overline{A} \coloneqq \Sigma^* \setminus A$.

\ifblock{The domain and image of a function $f$ are denoted as $\dom (f)$ and $\img (f)$. The support of $f$ is defined as $\supp (f) \coloneqq \{x \in \dom (f) \mid f(x) \not = 0\}$. A function $f$ is injective on its support, if, for any $a,b \in \supp (f)$, $f(a)=f(b)$ implies $a=b$. A function $f'$ is an extension of a function $f$, denoted as $f \sqsubseteq f'$, if $f(x) = f'(x)$ for all $x \in \dom (f)$. If $x \notin \dom (f)$, then $f \cup \{x \mapsto y\}$ denotes the extension $f'$ of $f$ such that $f'(z) = f(z)$ for $z \not = x$ and $f'(x)=y$. For a (not necessarily injective) function $f$ the function $f^{-1}$ denotes its inverse function, i.e., for $y \in \img (f)$ it holds that $f(f^{-1}(y))=y$.}{}

\ifblock{}{The image of a function $f$ is denoted as $\img (f)$.} Let $\langle \cdot \rangle \colon \bigcup _{i \geq 0} \N^i \to \N$ be an injective polynomial time computable and invertible pairing function such that $|\langle u_1, \dots , u_n \rangle | = 2(|u_1| + \cdots + |u_n| + n)$. The logarithm to the base $2$ is denoted as $\log$. Furthermore, we define polynomial functions $p_i \colon \N \to \N$ for $i \in \N^+$ by $p_i(x) \coloneqq x^i+i$. \ifblock{The function computing the maximum element of a finite subset of $\N$ is denoted as $\max$.}{}

We use the default model of a Turing machine in the deterministic as well as in the non-deterministic variation, abbreviated by DTM and NTM respectively. The language decided by a Turing machine $M$ is denoted as $L(M)$. For a number $s \in \N$ the language of words that are accepted by a Turing machine $M$ in $s$ computation steps is denoted as $L^s(M)$. We use Turing transducer to compute functions. For a Turing transducer $F$ we
write $F(x)=y$ when on input $x$ the transducer outputs $y$. A Turing transducer $F$ computes a total function and we sometimes refer to the function computed by $F$  as ``the function $F$''. Let $\{F_i\}_{i \in \N}$ and $\{G_i\}_{i \in \N}$ be standard enumerations of polynomial time
Turing transducers. Let $\{N_i\}_{i \in \N}$ be a standard enumeration of
non-deterministic polynomial time Turing machines with the special property that
$N_0$ is the machine that always rejects and $N_1$ is the machine that always
accepts, that is $L(N_0) = \emptyset$ and $L(N_1) = \N$. The runtime of $F_i$,
$G_i$ and $N_i$ is bounded by $p_i$. 
\begin{proposition}\label{prop:universalTM}
    There is a Turing machine $M$ and a Turing transducer $F$ such that for all
    $i,s,x\in \N$ the following properties hold:
    \begin{itemize}
    \item $\langle i,x,0^s\rangle \in L(M) \Leftrightarrow x\in L^s(N_i)$
    \item $F(\langle i,x,0^s\rangle) \coloneqq \begin{cases}
            \langle 1, F_i(x) \rangle &\text{if } F_i(x) \text{ stops within
            } s \text{ steps}\\
            \langle 0,0\rangle &\text{else}
        \end{cases}$
 
    \item Both machines run in time $O(|i|s \log s)$.
    \end{itemize}
\end{proposition}

$\FP$, $\P$, and $\NP$ denote standard complexity classes \cite{pap81}. For a class $\mathcal{C}$ define $\mathrm{co}\mathcal{C} \coloneqq \{A \subseteq \Sigma^* \mid \overline{A} \in \mathcal{C}\}$. We define the Boolean hierarchy over $\NP$ inductively. Let $\mathcal{C}$ and $\mathcal{D}$ be arbitrary complexity classes. First, we define boolean operators on classes:
\begin{align*}
\mathcal{C} \land \mathcal{D} \coloneqq&\  \{A \cap B \mid A \in \mathcal{C} \land B \in \mathcal{D}\}\\
\mathcal{C} \lor \mathcal{D} \coloneqq&\  \{A \cup B \mid A \in \mathcal{C} \lor B \in \mathcal{D}\}
\end{align*}
Then $\BH _1 \coloneqq \NP$, $\BH_{2k} \coloneqq \coNP \land \BH _{2k-1}$, $\BH_{2k+1} \coloneqq \NP \lor \BH_{2k}$, and $\BH \coloneqq \bigcup _{k\geq 1} \BH_k$ where $\BH _2$ is called $\DP$ and $\BH$ is called Boolean hierarchy over $\NP$. We want to emphasize that $\DP = \NP \land \coNP$ and $\coDP = \NP \lor \coNP$. Wagner \cite{wag87} showed that $\BH_k \subseteq \BH_{k+1}$ and $\BH_k \subseteq \coBH_{k+1}$. The classes $\P^{\NP [k]}$ for $k \in \N ^+$ contain all sets that can be accepted by a polynomial time Turing machine that queries at most $k$ elements from an $\NP$-set. The resulting hierarchy $\P^{\NP[1]}, \P^{\NP [2]}, \dots$ is called bounded query hierarchy over $\NP$. Beigel \cite{bei91} shows that $\BH _{2^k-1} \cup \coBH_{2^k-1} \subseteq \P^{\NP[k]} \subseteq \BH_{2^k} \cap \coBH _{2^k}$.

We use the common polynomial time many-one reducibility for sets $A,B \subseteq \Sigma ^*$, i.e., $A \leqmp B$ if there exists an $f \in \FP$ such that $x \in A \Leftrightarrow f(x) \in B$. For a class $\mathcal{C}$ and some problem $A$, we say that $A$ is hard for $\mathcal{C}$ if for all $B \in \mathcal{C}$ it holds $B \leqmp A$. The set $A$ is called complete for $\mathcal{C}$ if $A \in \mathcal{C}$ and $A$ is hard for $\mathcal{C}$. We define the following complete problems for $\NP$ and $\DP$.
\begin{align*}
\NPC \coloneqq&\ \{\langle 0^i, x, 0^p \rangle \mid i \in \N, x \in \Sigma ^* \mbox{ and } x \in L^p(N_i)\}\\
\DPC \coloneqq&\ \{\langle 0^i, 0^j, x, 0^p \rangle \mid i,j \in \N, x \in \Sigma ^* \mbox{ and } x \in L^p(N_i) \cap \overline{L^p(N_j)}\}\\
\DPC ' \coloneqq&\ \DPC \cup \{w \mid \nexists i,j \in \N, x \in \Sigma ^* \colon w = \langle 0^i, 0^j, x, 0^p \rangle\}
\end{align*}
It is easy to see that $\NPC$ is $\NP$-complete and $\DPC$ and $\DPC '$ are $\DP$-complete. Furthermore, their complements are complete for $\coNP$ and $\coDP$ respectively. The purpose of $\DPC '$ is that $\overline{\DPC '}$ consists only of words of the form $\langle 0^i, 0^j, x, 0^p \rangle$, which simplifies some arguments in section \ref{sec:implication}. Let $N_{\NPC}$ denote the polynomial time machine with $L(N_{\NPC}) = \NPC$.

We use proof systems for sets defined by Cook and Reckhow \cite{cr79}. They define a function $f \in \FP$ to be a proof system for $\img (f)$. Furthermore:
\begin{itemize}
\item A proof system $g$ is (p-)simulated by a proof system $f$, denoted by $g \leq f$ (resp., $g \psim f$), if there exists a total function $\pi$ (resp., $\pi \in \FP$) and a polynomial $p$ such that $|\pi(x)| \leq p(|x|)$ and $f(\pi(x)) = g(x)$ for all $x \in \Sigma^*$. In this context the function $\pi$ is called simulation function. Note that $g \psim f$ implies $g \leq f$.
\item A proof system $f$ is (p-)optimal for $\img(f)$, if $g \leq f$ (resp., $g \psim f$) for all $g \in \FP$ with $\img (g) = \img (f)$.
\item A complexity class $\mathcal{C}$ has (p-)optimal proof systems, if every $A \in \mathcal{C}$ with $A \not = \emptyset$ has a (p-)optimal proof system.
\item We say that (p-)optimal proof systems translate from a class $\mathcal{C}$ to $\mathcal{D}$ if the existence of (p-)optimal proof systems for $\mathcal{C}$ implies their existence for $\mathcal{D}$.
\end{itemize}
By the following result of Köbler, Messner and Torán \cite{kmt03}, we can prove or refute the existence of (p-)optimal proof systems for a class $\mathcal{C}$ by proving or refuting the existence of such proof systems for a complete set of $\mathcal{C}$.
\begin{proposition}[\cite{kmt03}]\label{prop:closure reducibility}
If $A \subseteq \Sigma^*$ has a (p-)optimal proof system and $\emptyset \not = B \leqmp A$, then $B$ has a (p-)optimal proof system.
\end{proposition}
\begin{corollary}\label{cor:closure reducibility}
If $A \subseteq \Sigma ^*$ is a hard set for some class $\mathcal{C}$ and $A$ has a (p-)optimal proof system, then $\mathcal{C}$ has (p-)optimal proof systems.
\end{corollary}

Furthermore, it was shown by Köbler, Messner, and Torán \cite{kmt03} that the class of sets having (p-)optimal proof systems is closed under intersection. This result can easily be extended to the operator $\land$ for complexity classes.
\begin{proposition}[\cite{kmt03}]\label{prop:closure intersection}
If $A,B \subseteq \Sigma^*$, $A \cap B \not = \emptyset$ and both sets have a (p-)optimal proof system, then $A \cap B$ has a (p-)optimal proof system.
\end{proposition}
\begin{corollary}\label{cor:closure logic and}
If two classes $\mathcal{C}$ and $\mathcal{D}$ have (p-)optimal proof systems, then $\mathcal{C} \land \mathcal{D}$ has (p-)optimal proof systems.
\end{corollary}
\ifblock{
\begin{proof}
Let $A \in \mathcal{C} \land \mathcal{D}$. There are $B \in \mathcal{C}$ and $B' \in \mathcal{D}$ such that $A = B \cap B'$. By the premise of the claim, $B$ and $B'$ have (p-)optimal proof systems. Then by Proposition~\ref{prop:closure intersection}, there is a (p-)optimal proof system for $A$.
\end{proof}
}{}

Finally, every (p-)optimal proof system can be transformed into a (p-)optimal
proof system that runs in linear time by polynomially padding the proofs.
\begin{proposition}\label{prop:lineartimepfs}
    If $f$ is a (p-)optimal proof system for $A \subseteq \Sigma$, then there is
    a (p-)optimal proof system $g$ for $A$ that runs in linear time.
\end{proposition}

\section{Translation of P-Optimal Proof Systems from DP to coDP}\label{sec:implication}

In this chapter we show that p-optimal proof systems for $\mathrm{DP}$ imply p-optimal proof systems for $\mathrm{coDP}$. This proof is based on machine simulation which is a relativizable proof technique. Thus, the following theorem also holds in the presence of an arbitrary oracle $O$. 

\begin{theorem}\label{thrm:C2popt}
  If there exists a p-optimal proof system for $\DPC$, then there exists
    a p-optimal proof system for $\overline{\DPC '}$.
\end{theorem}

We start by sketching the key idea used in the proof. Our approach needs
some technique to verify that a given instance is in
$\overline{\DPC '}$. There is no known way to decide $\overline{\DPC '}$ in
polynomial time, but we can use the p-optimal proof system for $\DPC$ for this verification.
We define a function $f': \mathbb{N}\times \mathbb{N} \rightarrow \mathbb{N}$ such that there is a polynomial-time-computable encoding $c: \mathbb{N}\rightarrow  \mathbb{N}$ with $f'(a,c(x)) \in
\DPC$ if and only if $F_a(x) \in \overline{\DPC '}$ for all $a \in \mathbb{N}$
and $x \in \mathbb{N}$. Furthermore, $f'(a,x)$ can be computed in time
$|x|^{O(a)}$. We derive a class of functions $\{f'_a\}_{a\in
\mathbb{N}}$ from $f'$ by fixing the first input to $a$. Note that $f'_a$ runs
in polynomial time for a fixed $a\in \mathbb{N}$ and that $f'_a$ is a proof
system for $\DPC$ if and only if $F_a$ is a proof system for $\overline{\DPC '}$.
Now, we define a machine that uses an additional input to verify $F_a(x) \in
\overline{\DPC '}$. The inputs of the machine are $a,x,b$ and it accepts if
and only if $f(F_b(x)) = f'_a(x)$ for a p-optimal proof system $f$ of $\DPC$. So, if $F_a$ is a proof system for
$\overline{\DPC '}$, we know $f'_a$ is a proof system for $\DPC$. Thus, by the
fact that $f$ is p-optimal, there is a $b\in \mathbb{N}$ such that $f(F_b(x))
= f'_a(x)$ for all $x\in \mathbb{N}$. Thus, when knowing the value $b$, the
machine can verify $F_a(x)\in \overline{\DPC '}$ for all $x\in \mathbb{N}$ for
a proof system $F_a(x)$ for $\overline{\DPC '}$ by accepting $f(F_b(c(x))=f_a'(c(x))$. On the other hand if
$F_a(x)\notin \overline{\DPC '}$ there is no $b$ such that $f(F_b(c(x)))
= f_a'(c(x))$ because $f_a'(c(x)) \notin \DPC = \img (f)$.
\begin{proof}
We start by defining an NTM $A$ that
checks for given $a,y'$ whether $F_a(y') = \langle i,j,x',0^p\rangle \in
\overline{\DPC '}$. Since a deterministic polynomial-time computation cannot check every possible path of a $\mathrm{coNP}$ machine
$N_j$, $A$ gets a path $y$ of $N_j$ as an additional input and has the property
that it accepts for all possible paths $y$ if and only if $F_a(y') \in
\overline{\DPC '}$. Arguing over all $y'$ this means if $F_a$ is a proof system
for $\overline{\DPC '}$,  then for all $y'$ and all corresponding paths $y$
the machine $A$ accepts on input $a,y',y$.

  Let $f$ be a p-optimal proof system for $\DPC$. Without loss of generality we
    assume $f(x)$ can be computed in $O(|x|)$ time by
    Proposition~\ref{prop:lineartimepfs}. We define $A$ on input
    $x$ as follows.
  \begin{enumerate}[(i)]
      \item Check whether $x=\langle a,y,y' \rangle $ for some $a\in \mathbb{N}$ and $y,y' \in \Sigma^*$, otherwise reject.
      \item Check whether $F_a(y') = \langle i,j,x',0^p\rangle$ with $i,j,p\in \mathbb{N}$ and $x'\in \Sigma^*$ and whether $y\in\Sigma^p$, otherwise reject.
      \item Accept if $N_{i}(x')$ does not accept on path $y$ within $p$ steps.
      \item Simulate the first $p$ steps of $N_j(x')$.
      \item Accept if the simulation of $N_j(x')$ accepted within the first $p$
          steps, otherwise reject.
  \end{enumerate}
  \begin{observation}
      $A(x )$ runs in time $O(|x|^{3a})$ for $ x = \langle a,y,y'\rangle $.
  \end{observation}
  \begin{proof}
      Checking whether $x$ has the required format is possible in linear time.
      By Propositon~\ref{prop:universalTM}, $F_a(y')$ can be computed in time \begin{align*}
          O(|a|\cdot (|y'|^a+a) \log (|y'|^a+a))&\subseteq O(|a|\cdot
          (|y'|^{a}+a)^2)\subseteq O((|x |^a+a)^{3}) \\
          &\subseteq O(|x |^{3a}+a^3) \subseteq O(| x|^{3a}).
      \end{align*}
      The first $p$ steps of the path $y$ of $N_i(x')$ can be simulated in time
      \begin{align*}
          O(|i|\cdot p \log p)\subseteq O(|F_a(y')| \cdot |F_a(y')|^2 ) \subseteq
      O(|x |^{3a} )
      \end{align*}
      for all $a\in \mathbb{N}$. The computation $N_j(x')$ can be simulated
      in $O(|j|\cdot p\log p)\subseteq O(|x |^{3a})$ time for all $a\in \mathbb{N}$. 
  \end{proof}
  \begin{claim}\label{claim:machineA}
    Let $F_a(y') = \langle i,j,x',0^{p}\rangle $. Then $F_a(y') \in \overline{\DPC '} \Leftrightarrow \forall y\in \Sigma^p :\langle a,y,y'\rangle \in L(A)$.
  \end{claim}
  \begin{claimproof}
    First we show ''$\Rightarrow$``. We consider two cases:
    \begin{itemize}
        \item Suppose $x' \in L^p(N_{i})$. By $\langle i,j,x',0^p\rangle \in \overline{\DPC '}$, it holds $x'\in
            L^p(N_j)$. The machine $A$ on input $\langle a, y,y'\rangle$ with
            $F_a(y') = \langle i,j,x',0^p\rangle \in \overline{\DPC '}$ and
            $y\in \Sigma^p$ rejects only, if the non-deterministic check in
            step (iv) fails. But this is impossible since $x' \in L^p(N_j)$.
        \item Suppose $x' \notin L^p(N_i)$. In this case $N_i(x')$ does not accept within $p$ steps for all paths
            $y\in \Sigma^{p}$. Thus, the machine accepts in step (iii).
    \end{itemize}
      Now, we show ''$\Leftarrow$``. Again, we distinguish two cases:
      \begin{itemize}
          \item Suppose $A(\langle a,y,y'\rangle)$ accepts in step (iii) for all $y\in
              \Sigma^p$. This implies that $x' \notin L^p(N_i)$. Thus, $\langle
              i,j,x',0^p\rangle \in \overline{\DPC '}$.

          \item If $A$ accepts but not in step (iii), we conclude it accepts in
              step (v). Hence, it holds $x'\in L^p(N_j)$ and $\langle
              i,j,x',0^p\rangle \in \overline{\DPC '}$.\qedhere
      \end{itemize}
  \end{claimproof}

We want to define functions $f_a$ in such a way that $f_a$ is a proof system for $\DPC$ if $F_a$ is a proof system for $\overline{\DPC '}$. For this, we can exploit the relationship of $A$ to proof systems of $\overline{\DPC '}$ shown in Claim \ref{claim:machineA}. Specifically, $f_a$ trusts that $A(\langle a,y,y' \rangle)$ accepts for specific $y$ and $y'$ (note that the accepting behavior of $A$ has influence on $\DPC$), which is equivalent to $F_a$ being a proof system for $\overline{\DPC '}$.

  Choose $a_1\in \mathbb{N}$ such that $N_{a_1} = A$. Let $k_A\in \mathbb{N}^+$ be a constant such that $A(x)$ runs in time $k_A|x|^{3a}+k_A$. Recall that $N_0$ always rejects. We define a function $f^*:\mathbb{N} \rightarrow \mathbb{N}$:
  \begin{align*}
      f^*(x) \coloneqq \begin{cases}
      \langle a_1,0,\langle a,y,y'\rangle,0^{k_A|x|^{3a}+k_A}\rangle &\text{if }x = \langle a,y,y'\rangle 0 \wedge F_a(y')=\langle i,j,x',0^p\rangle \\ &\phantom{\text{if}}\wedge\ y \in \Sigma^p\\
      f(x') &\text{if } x = x'1\\
      f(0) &\text{else }
      \end{cases}
  \end{align*}
  
  \begin{observation}\label{obs:timefa}
      $f^*(x)$ runs in $O(|x|^{3a})$ time for $x= \langle a,y,y'\rangle 0$. 
  \end{observation}
  \begin{proof}
      The case distinction for the first case is possible in $O(|x|^{3a})$ time because computing $F_a(y')$ is possible in that time. The output of the first case with exception of the unary runtime parameter $0^{k_A\cdot (|x|^{3a})+k_A}$ is possible in linear time. The unary runtime parameter can be computed in $O(|x|^{3a})$ time. The output of the other cases is possible in linear time.
  \end{proof}
  
  We define a function $f_a$ that is obtained from $f^*$ by fixing an index $a\in \mathbb{N}$ of a polynomial-time function $F_a$.
  \begin{align*}
      f_a(x) \coloneqq \begin{cases}
      f^*(\langle a,y,y'\rangle 0) &\text{if } x = \langle a, y,y'\rangle0 \\
      f(x') &\text{if } x = x'1\\
      f(0) &\text{else}
      \end{cases}
  \end{align*}
  \begin{observation}
      For $a\in \mathbb{N}$ and $y,y'\in \Sigma^*$ it holds that $f_a(\langle
      a,y,y'\rangle 0) = f^*(\langle a,y,y'\rangle 0)$.
  \end{observation}
  \begin{observation}\label{obs:timefaa}
      For a fixed $a\in \mathbb{N}$ the function $f_a(x)$ can be computed in polynomial time.
  \end{observation}
  \begin{proof}
      This follows from Observation~\ref{obs:timefa} and the linear runtime of
      $f$.
  \end{proof}
  
  \begin{claim}\label{claim:f_adetail}
      Let $a \in \mathbb{N}$ and $y'\in\Sigma^*$ such that $F_a(y') = \langle
      i,j,x',0^p\rangle \notin \overline{\DPC '}$. Then there is a $y\in
      \Sigma^p$, such that $f_a(\langle a,y,y'\rangle0) \notin \DPC$.
  \end{claim}
  \begin{claimproof}
      By Claim~\ref{claim:machineA} we conclude the existence of a $y\in
      \Sigma^p$ such that $A(\langle a,y,y'\rangle)$ rejects. This implies
      $f_a(\langle a,y,y'\rangle 0) = \langle a_1,0,\langle a,y,y'\rangle,
      0^{k_A|x|^{3a}+k_A}\rangle \notin \DPC$.
  \end{claimproof}
  \begin{claim}\label{claim:f_a}
      Let $a \in \mathbb{N}$. Then $\img (F_a) \subseteq \overline{\DPC '} \Leftrightarrow \img (f_a) =\DPC $.
  \end{claim}
\begin{claimproof}
    First we show ''$\Rightarrow$``.
  $\DPC \subseteq \img (f_a)$ holds because $\img (f) \subseteq \img (f_a)$ and $f$ is
    a proof system for $\DPC$. Let $x \in \Sigma^*$. In the bottom two cases of $f_a$ and $f^*$ it is easy
    to see $f_a(x) \in \DPC$ and $f^*(x)\in \DPC$. So we can assume $x = \langle a, y,y'\rangle 0$ with $F_a(y')
    = \langle i,j,x',0^p\rangle$ and $y\in \Sigma^p$. Since $\img (F_a) \subseteq
    \overline{\DPC '}$, it holds $\langle i,j,x',0^p\rangle \in \overline{\DPC
    '}$. By Claim~\ref{claim:machineA} we obtain that $\langle a,y,y'\rangle
    \in L(A)$ for all $y\in \Sigma^p$. Thus, $f_a(x)= f^*(\langle
    a,y,y'\rangle0)=\langle a_1,0,\langle
    a,y,y'\rangle,0^{k_A|x|^{3a}+k_A}\rangle \in \DPC$ since for all $y\in \Sigma^p$, $N_{a_1}(\langle a,y,y'\rangle )$ accepts, $N_0(\langle a,y,y'\rangle)$ rejects and $|\langle a,y,y'\rangle| \le |x|$.

    ''$\Leftarrow$`` follows directly as the contraposition of
    Claim~\ref{claim:f_adetail}.
\end{claimproof}

We want to define another NTM $B$ that checks for given $a, y'$ whether $F_a(y') \in \overline{\DPC '}$. To achieve this we use Claims \ref{claim:f_adetail} and \ref{claim:f_a}. $B$ checks for all $y \in \Sigma^p$ whether $f(F_b(\langle a,y,y'\rangle 0))=f_a(\langle a,y,y'\rangle0 )$ on input $a,b,y'$. So if $F_a$ is a proof system for $\overline{\DPC '}$, then there is a $b$ such that $B(a,b,y')$ rejects for all $y'$. Furthermore, if $F_a(y') \notin \overline{\DPC '}$, then $B(a,b,y')$ accepts for all $b$. $B(x)$ operates as follows.
\begin{enumerate}[(i)]
    \item Check whether $x = \langle a,b, y'\rangle$ for some $a,b\in \mathbb{N}$ and $y'\in \Sigma^*$, otherwise reject.
    \item Check whether $F_a(y') = \langle i,j,x',0^p \rangle$ for some $i,j,p \in \mathbb{N}$ and $x'\in \Sigma^*$, otherwise reject.
    \item Branch non-deterministically every $y \in \Sigma^{p}$.
    \item Accept if $f(F_b(\langle a,y,y'\rangle0)) \ne f_a(\langle
        a,y,y'\rangle0)$.
    \item Reject.
\end{enumerate}
  \begin{observation}\label{obs:runtimeB}
      $B(x)$ runs in time $ O(|x|^{9a^2b})$ for $x=\langle a,b,y'\rangle$. 
  \end{observation}
  \begin{proof}
      Checking whether the input is formatted correctly is possible in linear
      time. $F_a(y')$ can be computed in $O(|a|\cdot (|y'|^a+a)\log
      (|y'|^a+a))\subseteq O(|x|^{3a})$ time. We also observe that $|F_a(y')| \le
      |y'|^a+a$. In line (iv) it holds that $|F_a(y')| \ge 2p \ge 2|y|$ and in line 1 it holds that $2|a|+2|y'|+8 \le |x|$, and therefore,
      \begin{align*}
      |\langle a,y,y'\rangle 0| = 2(|a|+|y|+|y'|+3)+1\le |F_a(y')|+|x|\le |x|^a+a+|x|\le |x|^{a+1}+a.
      \end{align*}
      Thus, by Observation~\ref{obs:timefa}, computing $f_a(\langle a,y,y'\rangle 0)$ is possible in time \begin{align*}
          O((|x|^{a+1}+a)^{3a})\subseteq O((|x|^{a+2})^{3a})\subseteq O(|x|^{3a^2+6a})\subseteq O(|x|^{9a^2} ).
      \end{align*}
      The value $F_b(\langle a, y,y'\rangle 0)$ can be computed in time \begin{align*}
          O(|b|\cdot ((|x|^{a+1}+a)^b+b) \log ((|x|^{a+1}+a)^b+b))&\subseteq O(|b|\cdot ((|x|^{a+2})^b+b) \\&\phantom{\subseteq{}}\cdot \log ((|x|^{a+2})^b+b))\\&\subseteq O(|b|\cdot (|x|^{ab+2b}+b) \log (|x|^{ab+2b}+b))\\&\subseteq O(|b|\cdot (|x|^{ab+2b+1}) \log (|x|^{ab+2b+1}))\\&\subseteq O(|b|\cdot |x|^{2ab+4b+1})\subseteq O(  |x|^{2ab+4b+2}).
      \end{align*} In particular $|F_b(\langle a,y,y'\rangle 0 )| \in O(|x|^{2ab+4b+2})$ and hence the computation of $f(F_b(\langle a,y,y'\rangle 0))$ is possible in time $O(|x|^{2ab+4b+2})$. We simplify the sum of these runtimes.
          \[O(|x|^{3a} + |x|^{9a^2}+|x|^{2ab+4b+2} )\subseteq O(|x|^{9a^2b})\qedhere\]
  \end{proof}
\begin{claim}\label{claim:machineBsound}
    Let $a\in \mathbb{N}$ and $y'\in \Sigma^*$, such that $F_a(y')=\langle
    i,j,x',0^p\rangle \notin \overline{\DPC '}$. Then for all $b\in \mathbb{N}$
    it holds that $\langle a,b,y' \rangle \in L(B)$.
\end{claim}
\begin{claimproof}
    By Claim~\ref{claim:f_adetail} there is a $y\in \Sigma^p$ such that
    $f_a(\langle a,y,y'\rangle0) \notin \DPC = \img (f)$. Thus, $B$ accepts in
    step (iv), because $\img(f) = \DPC$.
\end{claimproof}

\begin{claim}\label{claim:machineB}
     Let $a \in \mathbb{N}$, such that $\img (F_a) \subseteq \overline{\DPC
     '}$. Then, there is some $b\in \mathbb{N}$, such that for all $y' \in \Sigma^*$ it holds that $\langle a,b,y' \rangle \notin L(B)$.
\end{claim}
\begin{claimproof}
  By Claim~\ref{claim:f_a} and Observation~\ref{obs:timefaa}, we know that $f_a$
    is a proof system for $\DPC$. Since $f$ is a p-optimal proof system for
    $\DPC$, there exists some $b\in \mathbb{N}$, such that for all $\hat{x} \in
    \Sigma^*$ it holds that $f_a(\hat{x}) = f(F_b(\hat{x}))$. Thus, the
    computation $B(\langle a,b,y'\rangle )$ cannot accept in step (iv) independent of $y'$. Hence, the machine rejects.
\end{claimproof}

Now, we define a function $g_{a,b}$ for every pair of possible proof system $F_a$ and possible simulation function $F_b$. Similarly to $f_a$, the function $g_{a,b}$ trusts that $B(\langle a,b,y' \rangle)$ accepts for all $y'$ (note that the accepting behavior of $B$ has influence on $\DPC$). If $F_a$ is a proof system for $\overline{\DPC '}$, then there is a $b\in \mathbb{N}$ such that $g_{a,b}$ is a proof system for $\DPC$ because the machine $B$ accepts on input $a,b,y'$ for all $y'\in \Sigma^*$. For $F_a(y')\notin \overline{\DPC '}$, we know there is no $b\in \mathbb{N}$ such that $B$ accepts on input $a,b,y'$. Hence, the corresponding output of $g_{a,b}$ is not in $\DPC$.

 Let $b_1$ be the index of the NTM $B$, that is $N_{b_1} = B$. Furthermore, let $k_B\in \mathbb{N}^+$ be a constant such that $B(x)$ runs in time $k_B |x|^{9a^2b} +k_B$ for all $x\in \Sigma^*$. Recall that $N_1$ always accepts. We define a function $g:\mathbb{N} \rightarrow \mathbb{N}$ whose input consists of two indices $a,b\in \mathbb{N}$ of polynomial-time functions $F_a,F_b$ and a proof $y'\in \Sigma^*$.
\begin{align*}
    g(x) \coloneqq \begin{cases}
    \langle 1,b_1,\langle a,b,y' \rangle ,0^{k_B|x|^{9a^2b} + k_B}\rangle
    &\text{if } x=\langle a,b,y' \rangle 0 \\
    f(x') &\text{if } x=x'1\\
    f(0) & \text{else}
    \end{cases}
\end{align*}
\begin{observation}\label{obs:timeg}
    $g(x)$ runs in $O(|x|^{9a^2b})$ time for $x = \langle a,b,y'\rangle 0$. 
\end{observation}
\begin{proof}
Checking whether the input is formatted correctly is possible in linear time. Furthermore, the output with exception of the last entry of the list can be computed in linear time. The string $0^{k_B|x|^{9a^2b} + k_B}$ can be computed in $O(|x|^{9a^2b})$ time. 
\end{proof}

We define a function $g_{a,b}:\mathbb{N}\rightarrow \mathbb{N}$ that is obtained from $g$ by fixing two indices $a,b\in \mathbb{N}$ of polynomial-time functions $F_a,F_b$.
\begin{align*}
    g_{a,b}(x) \coloneqq \begin{cases}
        g(\langle a,b,y' \rangle 0) &\text{if } x = \langle a,b,y'\rangle 0\\
        f(x') &\text{if } x = x'1\\
        f(0) &\text{else }
    \end{cases}
\end{align*}

\begin{observation}
    For $x = \langle a,b,y'\rangle 0$ it holds that $g(x) = g_{a,b}(x)$.
\end{observation}

\begin{observation}\label{obs:timegab}
    For fixed $a,b\in \mathbb{N}$ the function $g_{a,b}(x)$ can be computed in polynomial time.
\end{observation}
\begin{proof}
    This follows directly from Observation ~\ref{obs:timeg} and the linear
    runtime of $f$.
\end{proof}

\begin{claim}\label{claim:g_absound}
    Let $a\in \mathbb{N}$ and $y'\in \Sigma^*$ such that $F_a(y')=\langle
    i,j,x',0^p\rangle  \notin
    \overline{\DPC '}$. Then for all $b\in \mathbb{N}$ it holds that $g_{a,b}(\langle a,b,y'\rangle 0) \notin \DPC$.
\end{claim}
\begin{claimproof}
    By Claim~\ref{claim:machineBsound} we know that for all $b\in
    \mathbb{N}$ the computation $B(\langle a,b,y'\rangle)$ accepts. This
    implies $g_{a,b}(\langle a,b,y'\rangle 0) = \langle 1, b_1,\langle
    a,b,y'\rangle, 0^{k_B|x|^{9a^2b}+k_B}\rangle \notin \DPC$.
\end{claimproof}
\begin{claim}\label{claim:g_ab}
    Let $a\in \mathbb{N}$ such that $\img (F_a) \subseteq \overline{\DPC '}$.
    Then there is some $b\in \mathbb{N}$ with $\img (g_{a,b}) = \DPC$.
\end{claim}
\begin{claimproof}
   Choose $b\in \mathbb{N}$ according to Claim~\ref{claim:machineB}. Then $\DPC
    = \img(f) \subseteq \img (g_{a,b})$ because $f$ is a proof system for $\DPC$. Let $x\in \Sigma^*$. In the bottom two cases of $g_{a,b}$ and $g$ we have
    $g_{a,b}(x)\in \DPC$ and $g(x)\in \DPC$. So we can assume $x=\langle a,b,y'\rangle 0$ and $g_{a,b}(x)
    = \langle 1,b_1,\langle a,b,y'\rangle, 0^{k_B|x|^{9a^2b}+k_B}\rangle$. By
    the choice of $b$, it holds $\langle a,b,y'\rangle \notin L(B)$. By
    Observation~\ref{obs:runtimeB} and the choice of $k_B$, $B(\langle a,b,y'\rangle)$ runs in time
    $k_B|\langle a,b,y'\rangle| ^{9a^2b} +k_B \le k_B|x|^{9a^2b}+k_B$.
    Therefore, $\langle 1,b_1,\langle a,b,y'\rangle , 0^{k_B|x|^{9a^2b}+k_B}
    \rangle \in \DPC$ and hence $g_{a,b}(x)\in \DPC$. This shows
    $\img (g_{a,b})\subseteq \DPC$.
\end{claimproof}

Finally, we define the p-optimal proof system $h$ for $\overline{\DPC '}$. The key difficulty is that $h$ wants to output $F_a(y')$ for all $a$ and $y'$ using a short proof only when $F_a$ is a proof system for $\overline{\DPC '}$. To do this $h$ must be able to check this property efficiently. We can do this as follows: if $f(F_c(\langle a,b,y' \rangle)) = g_{a,b}(\langle a,b, y' \rangle)$, then we output $F_a(y)$ and otherwise some arbitrary word from $\overline{\DPC '}$. If $F_a(y')\notin \overline{\DPC '}$, we know that there is no $b\in \mathbb{N}$ such that the corresponding output of $g_{a,b}$ is in $\DPC$ and the check correctly fails and $F_a(y')$ is not outputted. By contraposition we observe that we output $F_a(y')$ only if it is in $\overline{\DPC '}$. Hence, $h$ is a proof system for $\overline{\DPC '}$. Lastly, we show that $h$ p-simulates all proof systems for $\overline{\DPC '}$. Let $F_a$ be an arbitrary proof system for $\overline{\DPC '}$. Then there is a $b\in \mathbb{N}$ such that $g_{a,b}$ is a proof system for $\DPC$. Let $c\in \mathbb{N}$ be such that $f$ p-simulates $g_{a,b}$ with the function $F_c$. So, for all $y'\in \mathbb{N}$ the function $h$ outputs $F_a(y')$ for the input to $h$ corresponding to $a,b,c,y'$. Also this input is short in $a,b,c,y'$ and
can be computed in polynomial time in these parameters.

Let $h':\mathbb{N}\rightarrow\mathbb{N}$ be a linear time proof system for
$\overline{\DPC '}$. We define a function $h:\mathbb{N}\rightarrow \mathbb{N}$.
\begin{align*}
    h(x) \coloneqq \begin{cases}
        \langle i,j,x',0^{p}\rangle &\text{if } x = \langle a,b,c,\langle a,b,y'\rangle 0,0^{k_B\cdot |\langle a,b,y'\rangle 0|^{9a^2b} + k_B},0^{|c|\cdot (|\langle a,b,y'\rangle 0|^c+c)^2}\rangle 0 \wedge \\ & f(F_c(\langle a,b,y'\rangle 0)) = g_{a,b}(\langle a,b,y'\rangle 0) \wedge \\ &F_a(y') = \langle i,j,x',0^{p}\rangle\\
        h'(x') & \text{if }x = x'1\\
        h'(0) & \text{else}
    \end{cases}
\end{align*}
\begin{observation}\label{obs:hruntime}
    $h(x)$ runs in time $O(|x|)$.
\end{observation}
\begin{proof} The bottom two cases are trivial. For the first case we observe
    that checking, whether the input is formatted correctly, can be done in
    linear time. The part
    $0^{k_B|\langle a,b,y'\rangle 0|^{9a^2b} +k_B}$ can be checked in linear
    time by iterated multiplication. The computation $F_a(y')$ can be simulated
    in $O(|a|\cdot (|y'|^a+a)\log (|y'|^a+a))\subseteq O(
    (|y'|^a+a)^3)\subseteq O(|y'|^{3a})) \subseteq O(|x|)$ time. The computation $f(F_c(\langle a,b,y'\rangle  0))$ can be simulated in $O(|c|\cdot (|\langle a,b,y'\rangle 0|^c+c)\log (|\langle a,b,y'\rangle 0|^c+c)) \subseteq O(|x|)$ time and $g_{a,b}(\langle a,b,y'\rangle 0)$ can be simulated in $O(|\langle a,b,y'\rangle 0|^{9a^2b})\subseteq O(|x|)$ time by Observation~\ref{obs:timeg}. The output $\langle i,j,x',0^p\rangle$ can be computed in $O(|x|)$ time because all of its elements have been computed in the steps analyzed above.
\end{proof}
\begin{claim}
$h$ is a proof system for $\overline{\DPC '}$.
\end{claim} 
\begin{claimproof}
  We have $h \in \FP$ by Observation \ref{obs:hruntime}. $\overline{\DPC '}\subseteq \img (h)$, since $\img (h') \subseteq \img (h)$ and $h'$ is
    a proof system for $\overline{\DPC '}$. We show $\img (h) \subseteq \overline{\DPC '}$ by contradiction. Assume that
    there exists $x \in \Sigma^*$ such that $h(x) \notin \overline{\DPC '}$. The
    last two cases in the definition of $h$ give values obviously in $\overline{\DPC '}$. Thus,
    we only look at the first case. In particular $F_a(y') = \langle
    i,j,x',0^p \rangle$ and $g_{a,b}(\langle a,b,y'\rangle 0) = f(F_c(\langle
    a,b,y'\rangle 0))$. The second implies directly $g_{a,b}(\langle
    a,b,y'\rangle 0))\in \img (f) = \overline{\DPC '}$. Since $h(x) = F_a(y')$ in this case, by assumption 
    $F_a(y')\notin \overline{\DPC '}$. By Claim~\ref{claim:g_absound} we conclude
    the contradiction
    $g_{a,b}(\langle a,b,y'\rangle 0) \notin \overline{\DPC '}$.
  \end{claimproof}

\begin{claim}\label{claim:kurzbeweis}
    Let $a \in \mathbb{N}$ with $\img (F_a) \subseteq \overline{\DPC '}$. Then there exist $b,c \in \mathbb{N}$, such that\\
    \centerline{
     $\forall y' \in \Sigma^*: F_a(y') = \langle i,j,x',0^p\rangle = h(\langle a,b,c,\langle a,b,y'\rangle 0,0^{k_B\cdot |y'0|^{9a^2b} + k_B},0^{|c|\cdot (|\langle a,b,y'\rangle 0|^c+c)^2}\rangle 0)$}\label{ali:314}
    
\end{claim}
\begin{claimproof}
  Claim~\ref{claim:g_ab} shows that there is some $b\in \mathbb{N}$ such that
    $\img (g_{a,b})=\DPC$. By Observation~\ref{obs:timegab} this $g_{a,b}$ is
    a proof system for $\DPC$. Since $f$ is p-optimal, there exists $c\in
    \mathbb{N}$ such that $f(F_c(x))=g_{a,b}(x)$ for all $x\in \Sigma^*$. Let
    $y'\in \Sigma^*$. From $\img (F_a) \subseteq\overline{\DPC '}$ it follows
    $F_a(y') = \langle i,j,x',0^p\rangle$ for suitable $i,j,x',p$. Hence, in
    Claim~\ref{claim:kurzbeweis} we are always in the first case of $h$. It follows \begin{align*}
  h(\langle a,b,c,\langle a,b,y'\rangle 0 ,0^{k_B\cdot |\langle a,b,y'\rangle 0|^{9a^2b} + k_B},0^{|c|\cdot (|\langle a,b,y'\rangle 0|^c+c)^2}\rangle 0) = \langle i,j,x',0^p\rangle.
  \end{align*}
  This shows Claim~\ref{claim:kurzbeweis}.
\end{claimproof}

Let $a \in \mathbb{N}$ be arbitrary such that $F_a$ is a proof system for
$\overline{\DPC '}$. Choose $b,c$ according to Claim~\ref{claim:kurzbeweis}. Then the following $z:\mathbb{N}\rightarrow\mathbb{N}$ shows $F_a \le^{\mathrm{p}} h$.
\begin{align*}
    z(y') \coloneqq \langle a,b,c,\langle a,b,y'\rangle 0,0^{k_B\cdot |y'0|^{9a^2b} + k_B},0^{|c|\cdot (|\langle a,b,y'\rangle 0|^c+c)^2}\rangle 0 
\end{align*}
By Claim~\ref{claim:kurzbeweis} it holds $F_a(y') = h(z(y'))$. The function $z$
can be computed in polynomial time, because $a,b,c \in \mathbb{N}$ and $k_B\in
\mathbb{N}^+$ are constant values for a fixed $F_a$.
This proves that $h$ is a p-optimal proof system for $\overline{\DPC '}$
\end{proof}
\begin{corollary}\label{cor:dptocodp}
    If $\mathrm{DP}$ has p-optimal proof systems, $\mathrm{coDP}$ has p-optimal proof systems.
\end{corollary}
\begin{proof}
    Since $\DPC \in \mathrm{DP}$, we obtain that there is a p-optimal proof
    system for $\DPC$. Theorem~\ref{thrm:C2popt} shows that it follows that
    there is a p-optimal proof system for $\overline{\DPC '}$. The
    language $\overline{\DPC '}$ is $\le_{\mathrm{m}}^{\mathrm{p}}$-hard for
    $\mathrm{coDP}$. By Corollary~\ref{cor:closure reducibility}, there are p-optimal proof systems for $\mathrm{coDP}$.
\end{proof}

\section{Oracle Construction}\label{sec:oracle construction}
The result from Corollary \ref{cor:dptocodp} naturally leads to the question if the existence of optimal proof systems for $\DP$ translate to optimal proof systems for $\coDP$. In this section we show that a proof for this translation cannot be relativizable, i.e., we cannot expect to find such a proof with the currently available means. We show this by constructing an oracle relative to which the following properties hold:
\begin{enumerate}[(i)]
\item[P1:] $\overline{\NPC}$ has p-optimal proof systems (implying p-optimal proof systems for $\coNP$).
\item[P2:] $\overline{\DPC}$ has no optimal proof systems (ruling out optimal proof systems for $\coDP$).
\end{enumerate}
The notation in this section is inspired by the notational framework from Ehrmanntraut, Egidy, and Glaßer \cite{eeg22}. We start by extending the definitions of section \ref{sec:prelim} and introducing some further definitions and notations specifically designed for oracle constructions.

\subparagraph{Oracle specific definitions and notations.} We relativize the concept of Turing machines and Turing transducers by giving them access to a write-only oracle tape as proposed in \cite{rst84}. We relativize complexity classes, proof systems, complete sets, machines and machine enumerations, reducibilities, and (p-)simulation from section \ref{sec:prelim} by defining them over machines with oracle access, i.e., whenever a Turing machine or Turing transducer is part of a definition, we replace them by an oracle Turing machine or an oracle Turing transducer. We indicate the access to some oracle $O$ in the superscript of said concepts, i.e., $\P^O$, $\NP^O$, $\FP^O$, $\dots$ for complexity classes, $M^O$ for a Turing machine or Turing transducer $M$, $\NPC ^O$, $\DPC ^O$ and $\DPC '^O$ for the defined complete sets, $\leq _m^{\text{p},O}$ for many-one reducibility and $\leq ^O$ (resp., $\leq ^{\text{p},O}$) for (p-)simulation. Let $N_{\NPC}$ denote the polynomial-time oracle Turing machine with $L(N_{\NPC}^O) = \NPC^O$ for any oracle $O$ such that for all $y \in \Sigma ^*$ the longest oracle query of $N_{\NPC}^O(y)$ has length $\leq |y|$. We sometimes omit the oracles in the superscripts, e.g., when sketching ideas in order to convey intuition, but never in the actual proof. All of the already mentioned propositions, theorems, corollaries, and claims hold relative to any oracle, because their proofs are relativizable.

For a deterministic machine $M$, an oracle $O$ and an input $x$, we define $Q(M^O(x))$ as the set of oracle queries of the computation $M^O(x)$. Is $M$ a non-deterministic machine and accepts some input $x$, then $Q(M^O(x))$ is the set of oracle queries of the leftmost accepting computation path. A word $w \in \Sigma ^*$ can be interpreted as a set $\{i \in \N \mid w(i) = 1\}$. During the oracle construction we often use words from $\Sigma ^*$ to denote partially defined oracles. Oracle queries for undefined words are negatively answered. Observe that $|w| \in w1$ and $|w| \notin w0$, i.e., $|w|$ is the element whose membership gets defined when extending $w$ by one bit. We say that a computation $M^w(x)$ is definite if all queries on all computation paths are $<|w|$. A computation definitely accepts (resp., rejects), if it is definite and accepts (resp., rejects). For a function $f$ we say $f(x)=y$ is definite if there is a machine $F$ computing $f$ relative to all oracles and the computation $F(x)=y$ is definite relative to all oracles. For a set $L$ we say $x \in L$ (resp., $x \notin L$) is definite if there is a machine $M$ deciding $L$ relative to all oracles and the computation $M(x)$ accepts (resp., rejects) definitely relative to all oracles. 

For the oracle construction, we need to assign infinitely many numbers to each number in $\N$. Thus, we define the following: let $e(0) \coloneqq 2$, $e(i) \coloneqq 2^{e(i-1)}$ for $i \in \N ^+$. The sets $H_m \coloneqq \{e(\langle m,n \rangle) + 1 \mid n \in \N\}$ are disjoint, contain countably infinite odd numbers each and $H_m \in \P$ for all $m \in \N$. Furthermore, when $n \in H_m$, then $n < n' < 2^n$ implies $n' \notin H_{m'}$ for $m' \in \N$.

\subparagraph{Preview of the construction.}
\begin{enumerate}[1]
\item Work towards P1: For all $a \in \N$, the construction tries to achieve that $\img (F_a) \not = \overline{\NPC}$ and thus, $F_a$ is no proof system for $\overline{\NPC}$. If this is not possible, we start to \textit{encode} the mappings of $F_a$ (i.e., on which input it gives which output) into the oracle. Thus, the final oracle will contain the encoded mappings of $F_a$ for almost all inputs, which makes simulating $F_a$ easy using oracle queries.
\item Work towards P2: For all $b \in \N$, the construction tries to achieve that $\img (G_b) \not = \overline{\DPC}$ and thus, $G_b$ is no proof system for $\overline{\DPC}$. If this is not possible, we fix some $m \in \N$, make sure that $z_m$ is a proof system for $\overline{\DPC}$ and show that $z_m$ cannot be simulated by $G_b$. The latter is achieved by diagonalizing against every simulation function $\pi$, i.e., we make sure that $G_b^O$ does not simulate $z_b^O$ via $\pi$.
\end{enumerate}
To these requirements we assign the following symbols representing tasks: $\tau _a^1, \tau _b^2$ and $\tau_ {b,i}^2$ for all $a,b,i \in \N$. The symbol $\tau_ a^1$ represents the coding or the destruction of $F_a$ as a proof system for $\overline{\NPC}$, $\tau _b^2$ represents the destruction of $G_b$ as a proof system for $\overline{\DPC}$ and $\tau _{b,i}^2$ represents the diagonalization of $G_b$ against simulation functions $\pi$ with $p_i$ as upper bound for the output length.

\subparagraph{Coding of mappings.} For the coding, we injectively define the code word $c(a,x,y) \coloneqq \c{a}{x}{y}$ for $a \in \N$, $x,y \in \Sigma ^*$. We call any word of the form $c(\cdot, \cdot, \cdot)$ a \textit{code word}. The purpose of a code word $c(a,x,y)$ is to encode the computation $F_a(x)=y$. We show some properties of code words.
\begin{claim}\label{claim:codeword} For all $a \in \N$, $x,y \in \Sigma^*$ it holds that
\begin{enumerate}[(i)]
\item $|c(a,x,y)| \notin H_m$ for any $m \in \N$.
\item for fixed $a \in \N$ computing and inverting $c(a,x,y)$ can be done in polynomial time with respect to $|x| + |y|$.
\item relative to any oracle $\ell (Q(F_a(x))) \leq |c(a,x,y)|/4$ and $\ell (Q(N_\NPC(y))) \leq |c(a,x,y)|/4$.
\item for every partial oracle $w \in \Sigma^*$, if $c(a,x,y) \leq |w|$, then $F_a^w(x)$ and $N_\NPC^w(y)$ are definite.
\end{enumerate}
\end{claim}
\begin{claimproof}
To (i): The output of $\langle \cdot \rangle$ has even length. Thus, code words have even length and elements from $H_m$ have odd length. To (ii): The function $\langle \cdot \rangle$ is polynomial-time computable and polynomial-time invertible and for fixed $a$ the length of $a + 4(|x|^a+a+|y|) + |x| + |y|$ is polynomially bounded in $|x|+|y|$. To (iii) and (iv): The number of symbols that are written by $F_a (x)$ and $N_C(y)$ on the oracle tape is bounded by $|x|^a+a + |y| \leq |c(a,x,y)|/4$.
\end{claimproof}

\subparagraph{Influencing $\boldsymbol{\overline{\DPC}}$.} We want to control the membership of some words to $\overline{\DPC}$ relative to an arbitrary oracle. For this, we will use  the language
\begin{align*}
A^O &\coloneqq  \{x \in \Sigma ^* \mid \card{O^{=|x|}} \geq 2\} \cup \{x \in \Sigma ^* \mid  \card{O^{=|x|}} = 0\} \\ 
&= \{x \in \Sigma^* \mid \card{O^{=|x|}} \not = 1\},
\end{align*}
which depends on an oracle $O$. The following claim shows how $A^O$ influences $\overline{\DPC^O}$:
\begin{claim}\label{claim:influencing} Let $O$ be an arbitrary oracle. Then
\begin{enumerate}[(i)]
\item $A^O \in \coDP^O$ via machines $N_{k_1}$ and $N_{k_2}$ with $A^O = L(N_{k_1}^O) \cup \overline{L(N_{k_2}^O)}$ for some $k_1,k_2 \in \N$. The combined runtime of $N_{k_1}$ and $N_{k_2}$ is upper bounded by $p_k$ for $k \coloneqq k_1+k_2$.
\item for all $n \in \N$ and for $k$, $k_1$, $k_2$ from (i) it holds: 
\[\card{O^{=n}} \not = 1 \iff \forall x \in \Sigma^n,\ \codecoDP{k_1}{k_2}{x}{p_k(|x|)} \in \overline{\DPC^O}.\]
\end{enumerate}
\end{claim}
\begin{claimproof}
To (i): We have $\{x \in \Sigma ^* \mid \card{O^{=|x|}} \geq 2\} \in \NP^O$ and $\overline{\{x \in \Sigma ^* \mid  \card{O^{=|x|}} = 0\}} = \{x \in \Sigma^* \mid \card{O^{=|x|}} \geq 1\} \in \NP^O$ via two polynomial-time machines $N_{k_1}^O$ and $N_{k_2}^O$, which shows $A^O \in \coDP^O$. For $k=k_1+k_2$ their combined runtime is bounded by $n^{k_1}+k_1 + n^{k_2} + k_2 \leq p_k(n)$.

To (ii): Let $n \in \N$ be arbitrary. By the definition of $A^O$, it holds that $\card{O^{=n}} \not = 1$ if and only if $x \in A^O$ for all $x \in \Sigma^{n}$. By (i), there are two machines $N_{k_1}$ and $N_{k_2}$ with $A^O = L(N_{k_1}^O) \cup \overline{L(N_{k_2}^O)}$ running in time $p_k$. Thus, by the definition of $\overline{\DPC^O}$, we obtain that $x \in A^O$ if and only if $\codecoDP{k_1}{k_2}{x}{p_k(|x|)} \in \overline{\DPC^O}$.
\end{claimproof}

For the rest of this chapter we fix $k$, $k_1$ and $k_2$ as the numbers described in Claim~\ref{claim:influencing}.(i). During the construction of the oracle $O$ we can influence the membership of words $\codecoDP{k_1}{k_2}{x}{p_k(|x|)}$ for arbitrary $x \in \Sigma ^*$ to the set $\overline{\DPC^O}$ by adding zero, one, or more words of length $|x|$ to $O$. Later, this will be the key mechanism to argue that our treatment of the tasks $\tau _{b,i}^2$ is possible.

\subparagraph{Witness proof systems.} As previously mentioned, for every proof system $G_b$ for $\overline{\DPC}$ we want to find some other proof system which is not simulated by $G_b$ such that $G_b$ cannot be an optimal proof system. This ``other proof system'' functions as a witness that $G_b$ is not optimal. Let $g^O$ be some arbitrary proof system for $\overline{\DPC ^O}$ for every oracle $O$, which can be computed in time $|x|$ for every input $x$ (observe that such $g$ exists). Then for every $m \in \N$ and all oracles $O$ we define $z_m^O$ as
\begin{equation*}
    z_m^O(x) \coloneqq
    \begin{cases*}
      g^O(x') & if $x = 1x' $ \\
      \codecoDP{k_1}{k_2}{x}{p_k(|x|)}  & if  $x = 0^{n}$ for $n \in H_m$ \\
      g^O(x)        &  else
    \end{cases*}.
\end{equation*}
\begin{claim}\label{claim:witnessfunction} Relative to every oracle $O$ and for every $m \in \N$:
\[ \forall n \in H_m,\ \card{O^{=n}} \not = 1 \Rightarrow z_m^O$ is a proof system for $\overline{\DPC^O}.\] 
\end{claim}
\begin{claimproof}
We have $z_m^O \in \FP^O$, because $g^O, p_k, \langle \cdot \rangle \in \FP^O$ and the case distinction can be done in polynomial time since $H_m \in \P$. Ignoring the second case of the definition of $z_m^O$, it holds that $\img (z_m^O) = \overline{\DPC^O}$, because $g^O$ is a proof system for $\overline{\DPC^O}$ and $z_m^O$ outputs the whole image of $g^O$.

Now we look at the second case of the definition. If $\card{O^{=n}} \not = 1$ for all $n \in H_m$, then by Claim~\ref{claim:influencing}.(ii), $\codecoDP{k_1}{k_2}{x}{p_k(|x|)} \in \overline{\DPC ^O}$ with $x=0^n$ for all $n \in H_m$. Thus, $z_m^O$ remains a proof system for $\overline{\DPC ^O}$.
\end{claimproof}

We call the functions $z_m^O$ witness proof systems. Each proof system for $\overline{\DPC^O}$ gets assigned some witness proof system $z_m^O$ during the oracle construction. The intuition behind the definition (and thus behind the whole diagonalization) is as follows: independent of the oracle each function $z_m^O$ has short proofs for $\codecoDP{k_1}{k_2}{0^n}{p_k(n)}$ with $n \in H_m$. Remember that the membership of these words to $\overline{\DPC^O}$ can be influenced by adding words of length $n$ to $O$. There are $2^n$ of such words. Let $G_b^O$ be some proof system for $\overline{\DPC^O}$. Either $G_b^O$ has a long proof and is able to know the number of words of length $n$ in $O$ before outputting $\codecoDP{k_1}{k_2}{0^n}{p_k(n)}$, but is unable to simulate $z_m^O$. Or $G_b^O$ also has a short proof, but cannot be sure that $\codecoDP{k_1}{k_2}{0^n}{p_k(n)} \in \overline{\DPC ^O}$ holds. Intuitively, we can construct $O$ such that $G_b^O$ cannot simulate $z_m^O$ and be a proof system for $\overline{\DPC^O}$ simultaneously. This describes the main idea of the diagonalization. Due to the necessity of code words in the oracle for property P1, realizing this approach gets more complicated.

\subparagraph{Valid Oracles.}
During the construction we successively add requirements that we maintain, which are specified by a partial function belonging to the set $\mathcal{T}$, and that are defined as follows: $t \in \mathcal{T}$ if $t$ partially maps $\tau _a^1, \tau _b^2$, $\tau _{b,i}^2$ to $\N$, and $\dom (t)$ is finite, and $t$ is injective on its support.

A partial oracle $w \in \Sigma ^*$ is called $t$-valid for $t \in \mathcal{T}$, if it satisfies the following requirements:
\begin{enumerate}[(i)]
	\item[V1] If $c(a,x,y) \in w$, then $y \in \overline{\NPC^w}$ is definite.
	
	(Meaning: the oracle contains only code words whose encoding of the ``output information'' is in $\overline{\NPC}$. This ensures that proof systems for $\overline{\NPC}$ can output $y$ when a code word $c(a,x,y)$ is inside the oracle. Note that $F_a^w(x)=y$ is not required.)

	\item[V2] If $t(\tau _a^1)=0$, then there exists an $x$ such that $F_a^w(x) \notin \overline{\NPC^w}$ is definite.
	
	(Meaning: if $t(\tau _a^1)=0$, then for every extension of $w$, $F_a$ is no proof system for $\overline{\NPC}$.)
	
	\item[V3] If $0< t(\tau _a^1) \leq c(a,x,y) < |w|$ and $F_a^w(x)=y$, then $c(a,x,y) \in w$.

	(Meaning: if $t(\tau _a^1)>0$, then from $t(\tau _a^1)$ on, we encode the mappings of $F_a$ into the oracle.)
	
	\item[V4] If $t(\tau _b^2)=0$, then there exists an $x$ such that $G_b^w(x) \notin \overline{\DPC^w}$ is definite.
	
	(Meaning: if $t(\tau _b^2)=0$, then for every extension of $w$, $G_b$ is no proof system for $\overline{\DPC}$.)
	
	\item[V5] If $m \coloneqq t(\tau _b^2)>0$, then for every $n \in H_m$ it holds that $\card{\Sigma ^n \cap w} \not = 1$.
	
	(Meaning: if $t(\tau _b^2)>0$, then for every extension of $w$ ensure that $\img(z_m) = \overline{\DPC}$ such that $z_m$ can be a witness proof system for $G_b$.)
	
	\item[V6] If $\tau _{b,i}^2 \in \dom (t)$, then for $m \coloneqq t(\tau _b^2)$ and every function $\pi$ with $|\pi (z)| \leq p_i(|z|)$ for all $z \in \Sigma^*$ there exists an $x$ such that $z_m^w(x) \not = G_b^w(\pi(x))$ and $z_m^w(x)$ and $G_b^w(\pi(x))$ are definite.
	
	(Meaning: if the task $\tau _{b,i}^2$ was treated, then for every extension of $w$ there is no function whose output length is bounded by $p_i$ that can be used by $G_b$ to simulate $z_m$.)
	
\end{enumerate}
The following observation follows directly from the fact, that extensions $t' \in \mathcal{T}$ of a function $t \in \mathcal{T}$ only add but never remove requirements.
\begin{observation}
Let $t,t' \in \mathcal{T}$ such that $t'$ is an extension of $t$. Whenever $w \in \Sigma ^*$ is $t'$-valid, then $w$ is $t$-valid. In particular, $w$ remains $t$-valid, even if $w$ contains code words $c(\cdot, \cdot, \cdot)$ that $t$ does not require to be in $w$, as long as $V1$ is satisfied.
\end{observation}

\subparagraph{Generic oracle extensions.} In the oracle construction and the subsequent proofs we extend partial oracles. Many of these extensions are done in a similar way. The following definition gives a short notation to refer to such extensions of partial oracles.

Let $u$ be a partial oracle, $n \in \N$ such that $u$ defines no word of length $n$, $m \in \N$ is an upper bound for the length of the extension, $t \in \mathcal{T}$ is some requirement function and $X \subseteq \Sigma ^n$. We construct $u_{t, n, m}(X)$ inductively. Basis clauses:
\begin{enumerate}[1]
\item If $z < |u|$, let $z \in u_{t,n,m}(X)$ if and only if $z \in u$.

(Meaning: $u_{t,n,m}(X)$ is an extension of $u$.)
\end{enumerate}
Inductive clauses: Let $z \geq |u|$ and $z \in \Sigma ^{\leq m}$.
\begin{enumerate}[1]\setcounter{enumi}{1}
\item If $z \in \Sigma ^n$, let $z \in u_{t,n,m}(X)$ if and only if $z \in X$.

(Meaning: $X$ defines the membership of words of length $n$, i.e., $u_{t,n,m}(X) \cap \Sigma ^n = X$.)
\item If $z=c(a,x,y)$, and $0 < t(\tau _a^1) \leq z$, and $F_a^{u_{t,n,m}(X)}(x)=y \in \overline{\NPC^{u_{t,n,m}(X)}}$, then let $z \in u_{t,n,m}(X)$.

(Meaning: If we are at a position of some mandatory code word (cf. V3), add the code word to the oracle if it does not violate V1.)

\item Otherwise, $z \notin u_{t,n,m}(X)$.
\end{enumerate}
If $t$, $n$ and $m$ are clear from the context, we may omit them in the subscript of $u(X)$. Intuitively, if $u$ is defined by the oracle construction, $\card{X} \not = 1$ and $u$ is $t$-valid, then $u_{t,n,m}(X)$ is a $t$-valid extension of $u$ up to words of length $m$ containing specific words of length $n$. We will prove this in Claim~\ref{claim:u(X)} after the oracle construction is defined.

\subparagraph{Dependency graph.} We want to construct $t$-valid oracles for some $t \in \mathcal{T}$. As can be seen in the definition of valid oracles, the membership of some code words to valid oracles may depend on computations. These computations can query the oracle and thus depend on words in the oracle. Hence, the membership of code words depends on words in the oracle. When arguing about the membership of code words to the oracle, we may need to know these dependencies. Thus, we capture these dependencies in a graph data structure.

Let $w$ be some oracle. Then $\mathcal{G}^w \coloneqq (V,E)$ is a directed graph defined as
\begin{align*}
V &\coloneqq \Sigma ^*,\\
E &\coloneqq \{(v,v') \mid v = c(a,x,y) \mbox{ for suitable } a,x,y \mbox{ and } v' \in Q(F_a^w(x))\}.
\end{align*}

So the graph considers all words in $\Sigma ^*$, edges only start at code words and edges represent oracle queries from the underlying $\FP$-computation. For a graph $\mathcal{G}^w = (V,E)$, $q \in V$ and $Q \subseteq V$ we define
\begin{align*}
R_{\mathcal{G}^w}(q) &\coloneqq \{v \mid \exists n \in \N, \exists v_1, \dots , v_n \in V,\ v_1=q \land v_n=v \land (v_i,v_{i+1}) \in E \mbox{ for } 1 \leq i < n\}\\
R_{\mathcal{G}^w}(Q) &\coloneqq \bigcup _{q \in Q} R_{\mathcal{G}^w}(q)
\end{align*}
as the set of nodes directly reachable from $q$ (resp., from all nodes in $Q$). For such graphs the following properties hold:
\begin{claim}\label{claim:graph}
For every oracle $w$ and every graph $\mathcal{G}^w=(V,E)$ it holds that:
\begin{enumerate}[(i)]
\item $\sum _{(v,v') \in E} |v'| \leq |v|/4$ for all $v \in V$.
\item $\mathcal{G}^w$ is a directed acyclic graph.
\item $\ell (R_{\mathcal{G}^w}(q)) \leq 2|q|$ for all $q \in V$.
\item $\ell (R_{\mathcal{G}^w}(Q)) \leq 2\ell(Q)$ for all $Q \subseteq V$.
\end{enumerate}
\end{claim}
\begin{claimproof}
To (i): A node $v \not = c(a,x,y)$ does not have outgoing edges. All directly reachable neighbors from $v = c(a,x,y)$ are in $Q(F_a^w(x))$. Thus, by Claim~\ref{claim:codeword}.(iii), the total length of reachable neighbors is upper bounded by $|v|/4$.

To (ii): From (i) follows that $|v| > |v'|$ holds for every edge $(v,v')$, making directed cycles impossible.

To (iii): We show the claim inductively on the height of nodes in $\mathcal{G}^w$, i.e., on the length of the longest path in $\mathcal{G}^w$ to a node with no outgoing edges. The height of each node in $\mathcal{G}^w$ is a natural number, because by (ii) $\mathcal{G}^w$ is acyclic and by (i) edges are only directed to nodes with smaller lexicographic representation. For nodes $v$ with height $0$ we have $\ell (R_{\mathcal{G}^w}(v)) = \ell (\{v\}) = |v|$. For nodes $v$ with height $> 0$ we get:
\begin{align*}
\ell (R_{\mathcal{G}^w}(v)) \leq |v| + \sum _{(v,v') \in E}\ell (R_{\mathcal{G}^w}(v')) \overset{\text{I.H.}}{\leq} |v| + \sum _{(v,v') \in E} 2|v'| \overset{\text{(i)}}{\leq} |v| + 2 |v|/4 \leq 2|v|
\end{align*}
To (iv): By (iii), we get:
\begin{align*}
\ell (R_{\mathcal{G}^w}(Q)) = \ell (\bigcup _{q\in Q}R_{\mathcal{G}^w}(q)) \leq \sum _{q \in Q} \ell (R_{\mathcal{G}^w}(q)) \overset{\text{(iii)}}{\leq} \sum _{q \in Q} 2|q| = 2 \ell (Q)
\end{align*}
\end{claimproof}

\subparagraph{Oracle construction.} Let $T$ be an enumeration of
\begin{align*}
\{\tau^1_{a} \mid a\in\N\} \cup \{\tau^2_{b} \mid b\in\N\} \cup \{ \tau^2_{b,i} \mid b,i\in\N\}
\end{align*} 
with the property that $\tau _b^2$ appears earlier than $\tau _{b,i}^2$.

We call $(w,t) \in \Sigma ^* \times \mathcal{T}$ a \textit{valid pair}, when $w$ is $t$-valid. The oracle construction inductively defines a sequence $\{(w_s,t_s)\}_{s\in \N}$ of \textit{valid pairs}. The $s$-th term is defined in stage $s$ of the oracle construction.

In each stage, we treat the smallest task in the order specified by $T$, and after treating a task we remove it and possibly other higher tasks from $T$. In the next stage, we continue with the next task not already removed from $T$. (In every stage, there always exists a task not already removed, as we never remove \textit{all} remaining tasks from $T$ in any stage.)

For stage $s=0$ we define $(w_0,t_0)$ with the nowhere defined function $t_0 \in \mathcal{T}$ and the $t_0$-valid oracle $w_0 \coloneqq \varepsilon$.

For stage $s > 0$ we have that $(w_0,t_0),\dots ,(w_{s-1},t_{s-1})$ are valid pairs. With this, we define $(w_s,t_s)$ such that $w_{s-1} \sqsubsetneq w_s$, and $t_{s-1} \sqsubseteq t_s$, and $w_s$ is $t_s$-valid, and the earliest task $\tau$ still in $T$ is treated and removed in some way.

Depending on whether $\tau _a^1$, $\tau _b^2$ or $\tau_{b,i}^2$ is the next task in $T$, the oracle construction is performed. We define the stage $s$ for all three cases:
\medskip

\textbf{Task} $\tau^1_{a}$: Let $t'\coloneqq t_{s-1}\cup \{\tau^1_{a} \mapsto 0\}$. If there exists a $t'$-valid $v\sqsupsetneq w_{s-1}$, then assign $t_s\coloneqq t'$ and $w_s \coloneqq v$.

Otherwise, let $t_s\coloneqq t_{s-1} \cup \{\tau^1_{a} \mapsto m \}$ with $m\in\N^+$ sufficiently large such that $m>|w_s|$ and $m>\max\img(t_{s-1})$. Thus, $t_s$ is injective on its support, and $w_{s-1}$ is $t_s$-valid, because V3 can only be violated for $\tau _a^1$ when $m < |w_{s-1}|$. Let $w_s\coloneqq w_{s-1}c$ with $c\in\{0,1\}$ such that $w_s$ is $t_s$-valid.
We will show in Lemma~\ref{lemma:extension} that such $c$ does indeed exist.

(Meaning: try to ensure that $F_a$ is no proof system for $\overline{\NPC}$ for all valid extensions of $w_{s}$ (cf. V2). If that is impossible, require that from $m$ on the computations of $F_a$ are encoded into the oracle (cf. V3).)
\medskip

\textbf{Task} $\tau^2_{b}$: Let $t'\coloneqq t_{s-1}\cup \{\tau^2_{b} \mapsto 0\}$. If there exists a $t'$-valid $v\sqsupsetneq w_{s-1}$, then assign $t_s\coloneqq t'$ and $w_s \coloneqq v$. Remove all tasks $\tau^2_{b,0}, \tau^2_{b,1}, \dots$ from $T$.

Otherwise, let $t_s\coloneqq t_{s-1} \cup \{\tau^2_{b} \mapsto m\}$ with $m\in\N^+$ sufficiently large such that $w_{s-1}$ defines no word of length $\min H_m$ and $m > \max\img(t_{s-1})$. Thus, $t_s$ is injective on its support, and $w_{s-1}$ is $t_s$-valid, because V5 can only be violated for $\tau _b^2$ when $w_{s-1}$ defines a word of length $\min H_m$. Let $w_s\coloneqq w_{s-1}c$ with $c\in\{0,1\}$ such that $w_s$ is $t_s$-valid. Again, we will show in Lemma~\ref{lemma:extension} that such $c$ does indeed exist.

(Meaning: try to ensure that $G_b$ is no proof system for $\overline{\DPC}$ for all valid extensions of $w_{s}$ (cf. V4). If that is impossible, choose a sufficiently large $m$ and require for the further construction that $\img (z_m) = \overline{\DPC}$ (cf. V5). The treatment of the tasks $\tau^2_{b,0}, \tau^2_{b,1}, \dots$ makes sure that $z_m$ cannot be simulated by $G_b$.)
\medskip

\textbf{Task} $\tau_{b,i}^2$: We have $t_{s-1}(\tau _b^2)=m \in \N^+$ (otherwise $\tau _{b,i}^2$ was removed from $T$). Let $t_s \coloneqq t_{s-1} \cup \{\tau_{b,i}^2 \mapsto 0\}$, and $n \in H_m$ sufficiently large such that $w_{s-1}$ defines no word of length $n$ and $2^{n-1} > 4p_i(n)^b+4b$. Extend $w_{s-1}$ to $w_s \coloneqq {w_{s-1}}_{t_{s-1},n,4p_i(n')^b+4b} (\emptyset)$. We will show in Proposition~\ref{proposition:diagonalization} that $w_s$ is $t_s$-valid. 

(Meaning: we want to prevent $G_b$ to simulate $z_m$. To rule out all simulation functions where the output length is bounded by $p_i$, we make sure, that there is a place in the oracle, that has no word of length $n$ for some $n\in H_m$ and $G_b$ is definite for inputs with length $\leq p_i(n)$.)
\bigskip

The list $t_s$ is always defined to be in $\mathcal{T}$. Remember that the treated task is immediately deleted from $T$. This completes the definition of stage $s$, and thus, the entire sequence $\{(w_s,t_s)\}_{s \in \N}$.

We now show that this construction is indeed possible, by stating and proving the lemma and proposition that was announced in the definition.

\begin{lemma}\label{lemma:extension}
Let $s \in \N$, $(w_0,t_0), \dots ,(w_{s-1},t_{s-1})$ be valid pairs defined by the oracle construction, and $t_s$ defined by the oracle construction, and let $w \in \Sigma ^*$ be a $t_s$-valid oracle with $w \sqsupseteq w_{s-1}$, and $z \coloneqq |w|$ is the next word we need to decide its membership to the oracle, i.e., $z \notin w0$ or $z \in w1$. Then there exists some $c \in \{0,1\}$ such that $wc$ is $t_s$-valid. Specifically:
\begin{enumerate}[(i)]
\item If $z = c(a,x,y)$, and $0 < t_s(\tau _a^1) \leq z$, and $F_a^{w}(x)=y \in \overline{\NPC^w}$, then $w1$ is a $t_s$-valid extension of $w$.

(Meaning: if we are at a position of some mandatory code word (cf. V3), add the code word to the oracle if it does not violate V1. This is equivalent to clause (3) in the definition of the generic oracle extension.)
\item Otherwise $w0$ is a $t_s$-valid extension of $w$.
\end{enumerate}
\end{lemma}
\begin{proof}
Assume that $wc$ is not $t_s$-valid. Then at least one of V1, V2, V3, V4, V5 and V6 has to be violated with respect to $t_s$. We show that a violation of any of these requirements leads to a contradiction.
\begin{itemize}
\item To V1: Since $w$ is $t_s$-valid, V1 holds for all words up to $z$. V1 can only be violated when $z$ is added to the oracle.  By (i), $z$ is only added when among other things $z=c(a,x,y)$ and $y \in \overline{\NPC^w}$. By Claim~\ref{claim:codeword}.(iv), $y \in \overline{\NPC^w}$ is definite. Thus, $y \in \overline{\NPC^{wc}}$ is definite and V1 is not violated.

\item To V2, V4 and V6: Since  V2, V4 and V6 are statements about the existence of definite computations which exist with respect to $w$, they also exist relative to any extension of $w$. Thus, V2, V4 and V6 cannot be violated.

\item Suppose V3 is violated. Then for some $a,x,y$ we have $0 < t_s(\tau _a^1) \leq c(a,x,y) < |wc|$, and $F_a^{wc}(x)=y$, and $c(a,x,y) \notin wc$. By Claim~\ref{claim:codeword}.(iv), we also have $F_a^w(x)=y$.

If $z \not = c(a,x,y)$, then $c(a,x,y) < |w|$. Thus, $0 < t_s(\tau _a^1) \leq c(a,x,y) < |w|$ and $F_a^{w}(x)=y$, but $c(a,x,y) \notin w$. Hence, V3 is already violated for $w$, which contradicts the $t_s$-validity of $w$.

Otherwise $z = c(a,x,y)$. Suppose $y \notin \overline{\NPC^w}$. Let $\hat{s}$ be the stage that treated $\tau _a^1$ and let $t' \coloneqq t_{\hat{s}-1} \cup \{\tau _a^1 \mapsto 0\}$. Observe $\hat{s} < s$. Since $w$ is $t_s$-valid, $w$ is also $t_{\hat{s}-1}$-valid. We have $F_a^w(x)=y \notin \overline{\NPC^w}$ definite and thus $w$ is $t'$-valid. In stage $\hat{s}$ of the oracle construction, the function $t_{\hat{s}-1}$ would have been extended to $t_{\hat{s}} \coloneqq t'$ with $t'(\tau _a^1) \coloneqq 0$, which is a contradiction to $t_s(\tau _a^1) = t_{\hat{s}}(\tau _a^1) > 0$.

Suppose $y \in \overline{\NPC^w}$. Then all conditions for case (i) are met and $w$ is extended to $w1$, contradicting our assumption of $c(a,x,y) \notin wc$.

\item Suppose V5 is violated. Then there is some $b,n \in \N^+$ with $t_s(\tau _{b}^2)=m$, $n \in H_m$ and $\card{\Sigma ^{n} \cap wc} = 1$. 

If $|z| \not = n$, then $\card{\Sigma^n \cap w} = \card{\Sigma ^n \cap wc } = 1$. This contradicts the $t_s$-validity of $w$.

Otherwise, by Claim~\ref{claim:codeword}.(i), no code word has length $n$. Thus, $z \not = c(a,x,y)$ for any $a,x,y$ and $w$ is extended by case (ii) to $w0$. Hence, $|\Sigma ^{n} \cap w| = |\Sigma ^{n} \cap w0| = |\Sigma ^n \cap wc| = 1$. This contradicts the $t_s$-validity of $w$.
\end{itemize} 
Since no requirement is violated for $wc$ with respect to $t_s$, $wc$ is a $t_s$-valid oracle.
\end{proof}

Lemma~\ref{lemma:extension} shows that the oracle construction is possible for the tasks $\tau _a^1$ and $\tau _b^2$. Now we show that the oracle construction is possible for the tasks $\tau _{b,i}^2$.
\begin{proposition}\label{proposition:diagonalization}
Let $s \in \N ^+$ be the stage that treats the task $\tau_{b,i}^2$, $(w_0,t_0),\dots ,(w_{s-1},t_{s-1})$ be valid pairs defined by the oracle construction, $m \coloneqq t_{s-1}(\tau _b^2)$, $t_s \coloneqq t_{s-1} \cup \{\tau _{b,i}^2 \mapsto 0\}$, and $w_s$ be the resulting oracle from the oracle construction. Then $w_s$ is $t_s$-valid.
\end{proposition}
\begin{proof}
Let us fix $b,i,m$ throughout the proof of the proposition. Also, let $\hat{s}$ be the stage that treated the task $\tau _b^2$, let $\gamma(n') \coloneqq 4(p_i(n')^b+b)$ be a polynomial function, and let $n \in H_m$ be the number from the oracle construction at stage $s$ with $2^{n-1} > \gamma(n)$. In this proof we are considering several extensions ${w_{s-1}}_{\cdot, \cdot, \cdot}(\cdot)$ of $w_{s-1}$. For a clearer notation let $u \coloneqq w_{s-1}$ and $u(X) \coloneqq u_{t_{s-1}, n, \gamma (n)}(X)$ for a set $X \subseteq \Sigma^n$, i.e., the oracle $w_{s-1}$ is abbreviated by $u$ and the parameters $t_{s-1}$, $n$ and $\gamma (n)$ are not mentioned explicitly in extensions $u(X)$ of $u$. Observe that $w_s = u(\emptyset)$.

We start by showing that $w_s$ is a $t_{s-1}$-valid oracle. Then it only remains to show that V6 for $\tau _{b,i}^2$ is satisfied with respect to $w_s$. 

\begin{claim}\label{claim:u(X)} Let $X \subseteq \Sigma ^n$.
\begin{enumerate}[(i)]
\item $u(X)^{<n} = u(Y)^{<n}$ for all $Y \subseteq \Sigma ^n$ and $u(X)^{<n}$ is $t_{s-1}$-valid.
\item If $\card{X} \not = 1$, then $u(X)$ is $t_{s-1}$-valid.
\item $u(X)$ is $t_{\hat{s}-1}$-valid.
\item If $\card{X} = 1$, then with respect to $t_{s-1}$ and $u(X)$, V5 is violated only for $\tau _b^2$ and V1, V2, V4 and V6 are not violated. 
\end{enumerate}
\end{claim}
\begin{claimproof}
To (i): Let $Y \subseteq \Sigma^n$ be arbitrary. The extensions $u(X)$ and $u(Y)$ only start to differ when they define the membership of some word $z$ by clause (2) from their definition. This can happen the first time for words $z$ with $|z| = n$. Hence, $u(X)^{<n} = u(Y)^{<n}$. Since $u$ is $t_{s-1}$-valid and $u(X)^{<n}$ is constructed exactly as described in Lemma~\ref{lemma:extension} with respect to $t_{s-1}$, $u(X)^{<n}$ is $t_{s-1}$-valid.

To (ii): By (i), $u(X)^{<n}$ is $t_{s-1}$-valid. Since $\card{X} \not = 1$, V5 is not violated for $u(X)^{\leq n}$ and thus also $t_{s-1 }$-valid. Then again, the remaining part of $u(X)$ is constructed as described in Lemma~\ref{lemma:extension} with respect to $t_{s-1}$ and thus $u(X)$ is $t_{s-1}$-valid.

To (iii): By (i), $u(X)^{<n}$ is $t_{\hat{s}-1}$-valid. Also $m \notin \dom (t_{\hat{s}-1})$. So adding any number of words of length $n \in H_m$ does not violate the $t_{\hat{s}-1}$-validity of $u(X)$, i.e., $u(X)^{\leq n}$ is $t_{\hat{s}-1}$-valid. Then, the remaining part of $u(X)$ gets extended exactly as described in Lemma~\ref{lemma:extension} with respect to $t_{s-1}$. Using Lemma~\ref{lemma:extension} with respect to $t_{s-1}$ instead of $t_{\hat{s}-1}$ leads to possibly some more code words $c(a,x,y)$ inside $u(X)$, because in clause (i) of Lemma~\ref{lemma:extension} the precondition $0 < t_{s-1}(\tau _a^1) \leq z$ may be satisfied but $0 < t_{\hat{s}-1}(\tau _a^1) \leq z$ may not be satisfied when $\tau_a^1$ is treated between the stages $\hat{s}-1$ and $s$. Only V1 can be violated when adding code words to $u(X)$. Since only code words $c(a,x,y)$ with $y \in \overline{\NPC^{u(X)}}$ are added and the membership is definite by Claim \ref{claim:codeword}.(iv), V1 is not violated with respect to $u(X)$. Thus, $u(X)$ is $t_{\hat{s}-1}$-valid.

To (iv): Since $t_{s-1}(\tau_b^2) =m > 0$ and $\card{\Sigma ^n \cap u(X)} = 1$, V5 is violated for $\tau_b^2$. By (ii), $u(X)^{<n}$ is $t_{s-1}$-valid and there is no $n' \in H_{m'}$ with $n \leq n' \leq \gamma (n)$ for any $m'\not = m$. 
Hence, V5 is satisfied for every other task $\tau _{b'}^2$ with $b' \not = b$. As argued in the previous cases, V1 cannot be violated. The requirements V2, V4 and V6 are about definite computations. Since $u$ is $t_{s-1}$-valid and $u \sqsubseteq u(X)$, these requirements also hold for $u(X)$.
\end{claimproof}

By Claim~\ref{claim:u(X)}.(ii) and $w_s=u(\emptyset)$, $w_s$ is $t_{s-1}$-valid. Since $t_s \coloneqq t_{s-1} \cup \{\tau_{b,i}^2 \mapsto 0\}$, $w_s$ is $t_s$-valid if V6 for $\tau_{b,i}^2$ is not violated.

\subparagraph{Remaining proof sketch.} We want to show that $z_m^{w_s}(0^n) = \codecoDP{k_1}{k_2}{0^n}{{p_k(n)}}$ and the shortest proof of $G_b^{w_s}$ for $\codecoDP{k_1}{k_2}{0^n}{{p_k(n)}}$ has length $>p_i(n)$. This proves $z_m^{w_s}(0^n) \not = G_b^{w_s}(\pi(0^n))$ for every simulation function $\pi$ whose output length is bounded by $p_i$. Hence, V6 is satisfied. We prove this by contradiction and assume that $G_b^{w_s}$ has a proof $x'$ of length at most $p_i(n)$. But then the oracle construction would have chosen $t_{\hat{s}}(\tau _b^2) \coloneqq 0$ in stage $\hat{s}$, i.e., we would have ruled out $G_b$ as a proof system for $\overline{\DPC}$. The main argument here is that we could add some word $z$ of length $n$ to the oracle (thus $\codecoDP{k_1}{k_2}{0^n}{{p_k(n)}} \notin \overline{\DPC^{u(\{z\})}}$) which is not asked by $G_b^{w_s}(x') = \codecoDP{k_1}{k_2}{0^n}{{p_k(n)}}$, since the ``shortness'' of $x'$ prevents $G_b^{w_s}(x')$ from asking many queries. The hard part is that the membership of code words $c(\cdot, \cdot, \cdot)$ of length $n<|c(\cdot, \cdot, \cdot)|\leq \gamma (n)$ may depend on $z$, which could in turn influence $G_b^{w_s}(x')$. So $z$ has to be chosen such that the computation $G_b^{w_s}(x')$ and also the membership of those code words, that influence $G_b^{w_s}(x')$, stay the same. Choosing $z$ is straightforward using $\mathcal{G}^{w_s}$, but proving the claimed properties is quite technical.

\bigskip

First observe that $w_s$ is defined for words up to length $\gamma (n)$ and thus for all inputs $x$ of length $\leq p_i(n)$ the computations $z_m^{w_s}(x)$ and $G_b^{w_s}(x)$ are definite. Also $z_m^{w_s}(0^n) = \codecoDP{k_1}{k_2}{0^n}{{p_k(n)}}$ by the definition of $z_m$ and $n \in H_m$. It remains to show that $G_b^{w_s}$ has no proof of length $\leq p_i(n)$ for $\codecoDP{k_1}{k_2}{0^n}{{p_k(n)}}$.

So assume that $G_b^{w_s}$ has a proof $x'$ for $\codecoDP{k_1}{k_2}{0^n}{{p_k(n)}}$ of length at most $p_i(n)$. Let 
\[Q_1 \coloneqq R_{\mathcal{G}^{w_s}}(Q(G_b^{w_s}(x'))\]
be the set of words the computation $G_b^{w_s}(x')$ (transitively) depends on. Then $\ell (Q(G_b^{w_s}(x'))) \leq \gamma (n)/4$ and $\ell (Q_1) \leq \gamma (n)/2$ follow from the bounded runtime of $G_b^{w_s}(x')$ and Claim~\ref{claim:graph}.(iv). Let $z \in \Sigma^n \setminus Q_1$. Such a word exists, because $2^{n-1} > \gamma (n)$. Under the stated assumption, the following claim holds, giving a contradiction to the result of the oracle construction.
\begin{claim}\label{claim:contradictingConstruction} The oracle construction would not create $t_{s-1}$. Specifically:
\begin{enumerate}[(i)]
\item If there is $a \in \N$ and $x,y \in \Sigma ^*$ with $0 < t_{s-1}(\tau_a^1) \leq c(a,x,y) < |u(\{z\})|$ and $F_a^{u(\{z\})}(x)=y \notin \overline{\NPC^{u(\{z\})}}$, then $t_{s-1}$ with $t_{s-1}(\tau _a^1) >0$ would not have been constructed.

(Meaning: if there is a function $F_a$ for which we have to encode its mappings into the oracle and it does not behave like a proof system for $\overline{\NPC}$ relative to $u(\{z\})$ (i.e., V3 is violated for some $c(a,x,y)$, because clause (3) in the definition of $u(\{z\})$ does not add $c(a,x,y)$), then we could have destroyed $F_a$ as a proof system for $\overline{\NPC}$ when treating the task $\tau_a^1$.)
\item Otherwise, $t_{s-1}$ with $t_{s-1}(\tau _b^2) > 0$ would not have been constructed.

(Meaning: if V3 is satisfied for all code words in $u(\{z\})$, then $u(\{z\})$ is almost $t_{s-1}$-valid, only V5 is violated for $\tau _b^2$ (see Claim \ref{claim:u(X)}.(iv)). Then we could have destroyed $G_b$ as a proof system for $\overline{\DPC}$ when treating the task $\tau_b^2$.)
\end{enumerate}
\end{claim}
\begin{claimproof}
To (i): Let $c(a,x,y)$ be the $\leq _{\text{lex}}$-minimal word that fulfills the requirements of case (i) and let $\hat{a}$ be the stage that treated the task $\tau _a^1$. We want to fix the computation path of $F_a^{u(\{z\})}(x)=y$ and the leftmost accepting path of $N_C^{u(\{z\})}(y)$ such that both also exist relative to some oracle $u(\{z,z'\})$ with $z' \in \Sigma ^n$ and $z' \not = z$. Together with Claim~\ref{claim:u(X)}.(ii), we would get that $u(\{z,z'\})$ is $t_{\hat{a}-1} \cup \{\tau _a^1 \mapsto 0\}$-valid, giving a contradiction to the construction of $t_{s-1}$ with $t_{s-1}(\tau_a^1)>0$. Let
\[Q_2 \coloneqq R_{\mathcal{G}^{u(\{z\})}}\left(Q(F_a^{u(\{z\})}(x)) \cup Q(N_C^{u(\{z\})}(y))\right)\]
be the words that these computation paths (transitively) depend on. According to Claim~\ref{claim:codeword}.(iii), it holds that $F_a(x)$ and $N_C(y)$ combined can write at most $|c(a,x,y)|/2$ symbols on the oracle tape, which gives 
\[\ell(Q_2) \leq 2|c(a,x,y)|/2 \leq \gamma (n).\]
Let $z' \in \Sigma^n \setminus Q_2$ with $z' \not = z$. The word $z'$ exists, because $2^{n}-1 > 2^{n-1} > \gamma (n)$. The following claim shows, that $z'$ is chosen such that no query of $F_a^{\smash{u(\{z\})}}(x)$ and $N_C^{\smash{u(\{z\})}}(y)$ gets answered differently relative to $u(\{z,z'\})$.
\begin{claim}\label{claim:equalanswers}
For all $q \in Q_2$ it holds that $q \in u(\{z\})$ if and only if $q \in u(\{z,z'\})$.
\end{claim}
\begin{claimproof}
Assume the claim does not hold. Let $q \in Q_2$ be the $\leq _{\text{lex}}$-minimal word with $q \in u(\{z\})$ if and only if $q \notin u(\{z,z'\})$. We look at several cases for $q$.

The case $|q| < n$ cannot occur, because by Claim~\ref{claim:u(X)}.(i), $u(\{z\})^{<n} = u(\{z,z'\})^{<n}$. Let $|q| = n$. Then only $q=z'$ is possible. But since $z'$ is chosen as $z' \notin Q_2$, this contradicts the premise $q \in Q_2$.

Let $|q| > n$. Either $q \in u(\{z\})$ or $q \in u(\{z,z'\})$. By the definition of $u(\{z\})$ and $u(\{z,z'\})$, $q$ has to be a code word $q = c(d,e,f)$ with $0<t_{s-1}(\tau_d^1) \leq q < |u(\{z\})|$. The following properties hold:
\begin{enumerate}[(i)]
\item $Q(F_d^{\smash{u(\{z\})}}(e)) \subseteq Q_2$. \hfill (since $q \in Q_2$ and nodes dir.\ reach.\ from $q$ in $\mathcal{G}^{u(\{z\})}$ are in $Q_2$)
\item $\ell (Q(F_d^{\smash{u(\{z\})}}(e))) < |q|$. \hfill (by Claim~\ref{claim:codeword}.(iii))
\item For all $q' \in Q(F_d^{\smash{u(\{z\})}}(e))$ it holds that $q' \in u(\{z\})$ if and only if $q' \in u(\{z,z'\})$.\\
(by (i), (ii) and the $\leq_{\text{lex}}$-minimality of $q$)
\item $F_d^{\smash{u(\{z,z'\})}}(e) = F_d^{\smash{u(\{z\})}}(e)$. \hfill (by (iii))
\end{enumerate}

\centerline{Suppose $q \in u(\{z\})$.}
\noindent Then $F_d^{\smash{u(\{z\})}}(e)=f \in \overline{\NPC^{u(\{z\})}}$ by the definition of $u(\{z\})$. We make another case distinction on whether $f$ is inside or outside of $\overline{\NPC^{u(\{z,z'\})}}$.
\begin{itemize}
\item Suppose $f \notin \overline{\NPC^{u(\{z,z'\})}}$. Let $\hat{d}$ be the stage in the oracle construction that treated the task $d$. It holds that $\hat{d}<s$ and that by Claim~\ref{claim:u(X)}.(ii), $u(\{z,z'\})$ is $t_{s-1}$-valid and thus $t_{\smash{\hat{d}-1}}$-valid. Then $u(\{z,z'\})$ is also $t_{\smash{\hat{d}-1}} \cup \{\tau _d^1 \mapsto 0\}$-valid, because $F_d^{\smash{u(\{z,z'\})}}(e) =f \notin \overline{\NPC^{u(\{z,z'\})}}$ is definite by Claim \ref{claim:codeword}.(iv) and V2 is satisfied for $\tau _d^1$. This is a contradiction to $t_{s-1}(\tau _d^1) > 0$, because in stage $\hat{d}$, the oracle construction would have preferred $t_{\smash{\hat{d}}}(\tau _d^1) \coloneqq 0$.

\item Suppose $f \in \overline{\NPC^{u(\{z,z'\})}}$. By the choice of $q$, (iv) and $F_d^{\smash{u(\{z\})}}(e)=f \in \overline{\NPC^{u(\{z\})}}$, clause (3) in the definition of $u(\{z,z'\})$ gives $q \in u(\{z,z'\})$, a contradiction.
\end{itemize}

\centerline{Suppose $q \notin u(\{z\})$.}
\noindent Then by the definition of $u(\{z\})$, we either have $F_d^{\smash{u(\{z\})}}(e) \not = f$ or $F_d^{\smash{u(\{z\})}}(e) = f \notin \overline{\NPC^{u(\{z\})}}$.
\begin{itemize}
\item Suppose $F_d^{\smash{u(\{z\})}}(e) \not = f$. By (iv), $F_d^{\smash{u(\{z,z'\})}}(e) \not = f$. Hence, $q \notin u(\{z,z'\})$ by definition, a contradiction.

\item Suppose $F_d^{\smash{u(\{z\})}}(e) = f \notin \overline{\NPC^{u(\{z\})}}$. Then $q$ fulfills all requirements of case (i) in Claim~\ref{claim:contradictingConstruction}. Since $q \in Q_2$ and thus $q <_{\text{lex}} c(a,x,y)$, $q$ would be the $\leq _{\text{lex}}$-minimal word chosen in the beginning of the proof of Claim~\ref{claim:contradictingConstruction}, a contradiction.
\end{itemize}
\end{claimproof}

Claim \ref{claim:equalanswers} gives $F_a^{\smash{u(\{z,z'\})}}(x) = F_a^{\smash{u(\{z\})}}(x) = y$ and $y \notin \overline{\NPC^{u(\{z,z'\})}}$, because all oracle queries $Q_2$ of the fixed computation paths are answered the same relative to $u(\{z,z'\})$. Recall that $\hat{a}$ is the stage that treated the task $a$. By Claim~\ref{claim:u(X)}.(ii), $u(\{z,z'\})$ is $t_{s-1}$-valid and since $\hat{a} < s$, also $t_{\hat{a}-1}$-valid. Then $u(\{z,z'\})$ is $t_{\hat{a}-1} \cup \{\tau _a^1 \mapsto 0\}$-valid, because $F_a^{\smash{u(\{z,z'\})}}(x) = y \notin \overline{\NPC^{u(\{z,z'\})}}$ is definite and thus V2 is satisfied for $\tau_a^1$. This is a contradiction to $t_{s-1}(\tau _a^1) > 0$, because in stage $\hat{a}$, the oracle construction would have preferred $t_{\hat{a}}(\tau_a^1) \coloneqq 0$. 
\medskip

To (ii): Similarly to (i), we first show that the oracle queries of $G_b^{w_s}(x')$ and $G_b^{u(\{z\})}(x')$ are answered equally. Then we will use this to derive a contradiction to the construction of $t_{s-1}$. Recall that $w_s = u(\emptyset)$.
\begin{claim}\label{claim:equalanswers2}
For all $q \in Q_1$ it holds that $q \in w_s$ if and only if $q \in u(\{z\})$.
\end{claim}
\begin{claimproof}
Assume the claim does not hold. Let $q \in Q_1$ be the $\leq _{\text{lex}}$-minimal word with $q \in w_s$ if and only if $q \notin u(\{z\})$. We look at several cases for $q$.

The case $|q| < n$ cannot occur, because by Claim~\ref{claim:u(X)}.(i), $w_s^{<n} = u(\{z\})^{<n}$. Let $|q| = n$. Then only $q = z$ is possible. But since $z$ is chosen as $z \notin Q_1$, this contradicts the premise $q \in Q_1$.

Let $|q| > n$. Either $q \in w_s$ or $q \in u(\{z\})$. By the definition of $w_s$ and $u(\{z\})$, $q$ has to be a code word $q = c(d,e,f)$ with $0<t_{s-1}(\tau _d^1) \leq q < |w_s|$. Then by essentially the same arguments as presented in case (i) the following holds:
\begin{enumerate}[(i)]
\item $Q(F_d^{w_s}(e)) \subseteq Q_1$. \hfill (since $q \in Q_1$ and nodes dir.\ reachable from $q$ in $\mathcal{G}^{w_s}$ are in $Q_1$)
\item $\ell (Q(F_d^{w_s}(e))) < |q|$. \hfill (by Claim~\ref{claim:codeword}.(iii))
\item For all $q' \in Q(F_d^{w_s}(e))$ it holds that $q' \in w_s$ if and only if $q' \in u(\{z\})$.\\
(by (i), (ii) and the $\leq _{\text{lex}}$-minimality of $q$)
\item $F_d^{u(\{z\})}(e)=F_d^{w_s}(e)$. \hfill (by (iii))
\end{enumerate}

\centerline{Suppose $q \in w_s$.}
\noindent Then $F_d^{w_s}(e)=f \in \overline{\NPC^{w_s}}$ by the definition of $w_s$. We make another case distinction on whether $f$ is inside or outside of $\overline{\NPC^{u(\{z\})}}$.
\begin{itemize}
\item Suppose $f \notin \overline{\NPC^{u(\{z\})}}$. Then the code word $c(d,e,f)$ satisfies the requirements of case (i) in Claim~\ref{claim:contradictingConstruction}, which is a contradiction to being in case (ii).

\item Suppose $f \in \overline{\NPC^{u(\{z\})}}$. By the choice of $q$, (iv) and $F_d^{w_s}(e)=f \in \overline{\NPC^{w_s}}$, clause (3) in the definition of $u(\{z\})$ gives $q \in u(\{z\})$, a contradiction.
\end{itemize}

\centerline{Suppose $q \notin w_s$.}
\noindent Then by the definition of $w_s$, we either have $F_d^{w_s}(e) \not = f$ or $F_d^{w_s}(e) = f \notin \overline{\NPC^{w_s}}$.
\begin{itemize}
\item Suppose $F_d^{w_s}(e) \not = f$. By (iv), $F_d^{u(\{z\})}(e) \not = f$. Hence, $q \notin u(\{z\})$ by definition, a contradiction.

\item Suppose $F_d^{w_s}(e) = f \notin \overline{\NPC^{w_s}}$. Let $\hat{d}$ be the stage in the oracle construction that treated the task $d$. It holds that $\hat{d} < s$ and thus, $w_s$ is $t_{\smash{\hat{d}-1}}$-valid. Then $w_s$ is also $t_{\smash{\hat{d}-1}} \cup \{\tau _d^1 \mapsto 0\}$-valid, because $F_d^{w_s}(e)=f \notin \overline{\NPC^{w_s}}$ is definite and V2 is satisfied for $\tau _d^1$. This is a contradiction to $t_{s-1}(\tau _d^1) > 0$, because in stage $\hat{d}$, the oracle construction would have preferred $t_{\smash{\hat{d}}}(\tau _d^1) \coloneqq 0$.
\end{itemize}
\end{claimproof}

Claim~\ref{claim:equalanswers2} gives $G_b^{\smash{u(\{z\})}}(x') = G_b^{w_s}(x') = \codecoDP{k_1}{k_2}{0^n}{{p_k(n)}}$, because all oracle queries $Q_1$ of the computation $G_b^{w_s}(x')$ are answered the same relative to $u(\{z\})$. With this we can derive a contradiction to the construction of $t_{s-1}$. By Claim~\ref{claim:u(X)}.(iii), $u(\{z\})$ is $t_{\hat{s}-1}$-valid. By Claim~\ref{claim:influencing}.(ii), $\codecoDP{k_1}{k_2}{0^n}{{p_k(n)}} \notin \overline{\DPC^{u(\{z\})}}$. Then $u(\{z\})$ is $t_{\hat{s}-1} \cup \{\tau _b^2 \mapsto 0\}$-valid, because $G_b^{\smash{u(\{z\})}}(x') = \codecoDP{k_1}{k_2}{0^n}{{p_k(n)}} \notin \overline{\DPC^{u(\{z\})}}$ is definite and thus V4 is satisfied for $\tau _b^2$. This is a contradiction to $t_{s-1}(\tau _b^2)>0$.
\end{claimproof}

Claim~\ref{claim:contradictingConstruction} shows that the assumption that $G_b^{w_s}$ has a proof $x'$ for $\codecoDP{k_1}{k_2}{0^n}{{p_k(n)}}$ leads to a different requirements list $t_{s-1}$ from the one received by the oracle construction, which is a contradiction. This proves the proposition.
\end{proof}

We have completed the proofs showing that the oracle construction can be performed as desired. The following theorem confirms the desired properties of $O \coloneqq \bigcup _{i\in \N}w_i$. Remember that $|w_0| < |w_1| < \dots $ is unbounded, hence for any $z$ there is a sufficiently large $s$ such that $|w_s|>z$. Remember that $(w_s,t_s)$ is a valid pair for all $s \in \N$ and all tasks initially in $T$ were treated or removed in some stage.
\begin{theorem}\label{theorem:oracleconstruction}
Relative to $O \coloneqq \bigcup _{i \in \N}w_i$, the following holds:
\begin{enumerate}[(i)]
\item $\overline{\NPC^O}$ has p-optimal proof systems.
\item $\overline{\DPC^O}$ has no optimal proof systems.
\end{enumerate}
\end{theorem}
\begin{proof}
To (i): Let $h'$ be an arbitrary proof system for $\overline{\NPC^O}$. Then the following proof system $h$ is a p-optimal proof system for $\overline{\NPC^O}$:
\begin{equation*}
   h(z) \coloneqq
    \begin{cases*}
      y & if $z = 0c(a,x,y)$ and $c(a,x,y) \in O$\\
      h'(z')  & if $z=1z'$\\
      h'(z)        &  else
    \end{cases*}
\end{equation*}
\begin{claim}
The function $h$ is a proof system for $\overline{\NPC^O}$.
\end{claim}
\begin{claimproof}
It holds that $h \in \FP^O$, because $h' \in \FP^O$ and the distinction of the cases can be done in polynomial time, because identifying a word as a code word can be done in polynomial time. 

Since $h'$ is a proof system for $\overline{\NPC^O}$ and $\img (h) \supseteq \img (h')$, we have $\img(h) \supseteq \overline{\NPC^O}$. The second and the third line in the definition of $h$ provide only elements inside $\overline{\NPC^O}$. If $c(a,x,y) \in O$, then $c(a,x,y) \in w_s$ for some $s \in \N$ with $c(a,x,y) < |w_s|$, because $O$ is an extension of $w_s$. Recall that $w_s$ is $t_s$-valid. By V1, $y \in \overline{\NPC^{w_s}}$ is definite and thus $y \in \overline{\NPC^O}$. Hence, the first line also outputs only elements inside $\overline{\NPC^O}$ which shows $\img(h) \subseteq \overline{\NPC^O}$.
\end{claimproof}

Now it remains to show that $h$ is p-optimal. Let $f$ be some arbitrary proof system for $\overline{\NPC^O}$. Then there is some Turing transducer $F_a^O$ which computes $f$. Let $\hat{a}$ be the stage that treated the task $\tau _a^1$. Then $t_{\hat{a}}(\tau _a^1) = m > 0$, because if $t_{\hat{a}}(\tau _a^1) = 0$, then by V2, $F_a^O$ would be no proof system for $\overline{\NPC^O}$. Let $h'^{-1}$ be the inverse function of $h'$. Consider the following function:
\begin{equation*}
    \pi(x) \coloneqq
    \begin{cases*}
      0c(a,x,f(x)) & if $m \leq c(a,x,f(x))$ \\
      1h'^{-1}(f(x))  & else
    \end{cases*}.
\end{equation*}
By $f(x) \in \img(h')$, $h'^{-1}(f(x))$ is always defined. By Claim~\ref{claim:codeword}.(ii), and the set $\{x \in \Sigma^* \mid c(a,x,f(x))<m\}$ being finite, and $f \in \FP^O$, we have $\pi \in \FP^O$. We now show $h(\pi(x))=f(x)$ for all $x \in \Sigma^*$. If $c(a,x,f(x)) < m$, then
\[h(\pi(x)) = h(1h'^{-1}(f(x))) = h'(h'^{-1}(f(x))) = f(x).\]
If $c(a,x,f(x)) \geq m$, then $c(a,x,f(x)) < |w_s|$ for some $s > \hat{a}$. Recall that $w_s$ is $t_s$-valid. By Claim~\ref{claim:codeword}.(iv), $F_a^{w_s}(x)=y$ is definite, thus $F_a^{w_s}(x)=F_a^O(x)=f(x)$. By V3, $F_a^{w_s}(x)=f(x)$ implies $c(a,x,f(x)) \in w_s$. Since $O$ is an extension of $w_s$, we have $c(a,x,f(x)) \in O$. This gives $h(\pi(x)) = h(0c(a,x,f(x))) = f(x)$.

This shows $f \leq^{p,O} h$ with $\pi$ as simulation function. Since $f$ was an arbitrary proof system for $\overline{\NPC^O}$, $h$ is a p-optimal proof system for $\overline{\NPC^O}$.
\medskip

To (ii): Assume otherwise that there is an optimal proof system $g$ for $\overline{\DPC^O}$. Let $G_b^O$ be the Turing transducer computing $g$. Let $s$ be the stage in the oracle construction that treated the task $\tau _b^2$ and let $m \coloneqq t_s(\tau _b^2)$. 

If $m = 0$, then V4 states that for some $x \in \Sigma ^*$, we have $G_b^{w_{s}}(x) \notin \overline{\DPC^{w_{s}}}$ definitely. Since $O$ is an extension of $w_s$, this also holds relative to $O$. Thus, $G_b^O$ is no proof system for $\overline{\DPC^O}$, a contradiction.

If $m > 0$, then consider $z_m^O$. V5 holds for $O$, because otherwise $w_{s'}$ would not be $t_{s'}$-valid for a sufficiently large $s' \in \N$. Thus, $\card{O^{=n}} \not = 1$ holds for all $n \in H_m$. Together with Claim~\ref{claim:witnessfunction}, we get that $z_m^O$ is a proof system for $\overline{\DPC^O}$. Let $G_b^O$ simulate $z_m^O$ via $\pi$. Then there is some polynomial $p_i$ bounding the output length of $\pi$, i.e., $|\pi(x)| \leq p_i(|x|)$ for all $x \in \Sigma ^*$. Let $\hat{s}$ be the stage in the oracle construction that treated the task $\tau _{b,i}^2$ and thus $\tau _{b,i}^2 \in \dom (t_{\hat{s}})$. Such a stage exists, because $t_s(\tau _b^2) = m > 0$. Consider the requirement V6 for the task $\tau _{b,i}^2$. From this requirement follows some $x \in \Sigma^*$ such that $z_m^{w_{\hat{s}}}(x) \not = G_b^{w_{\hat{s}}}(\pi (x))$ with $z_m^{w_{\hat{s}}}$ and $G_b^{w_{\hat{s}}}(\pi (x))$ being definite. Since $O$ is an extension of $w_{\hat{s}}$, this also holds relative to $O$. Thus, $z_m^O$ is not simulated by $G_b^O$ via $\pi$, a contradiction.
\end{proof}
\begin{corollary}\label{cor:dpnottocodp}
Relative to $O \coloneqq \bigcup _{i\in \N} w_i$, the following holds:
\begin{itemize}
\item $\coNP$ has p-optimal proof systems.
\item $\coDP$ has no optimal proof systems.
\item $\DP$ has optimal proof systems.
\item A translation of optimal proof systems from $\DP$ to $\coDP$ cannot be proved with relativizable proof techniques.
\end{itemize}
\end{corollary}

\section{Conclusion}
We summarize all results to obtain the equivalence classes from Figure \ref{fig:results}. First observe that (p-)optimal proof systems always translate from a class $\mathcal{C}$ to $\mathcal{D}$ when $\mathcal{C} \subseteq \mathcal{D}$ (respective solid arrows are omitted in Figure \ref{fig:conclusion}). We start with the equivalence classes for p-optimal proof systems (see Figure \ref{fig:conclusion}, left, solid arrows). P-optimal proof systems translate as follows: 
\begin{itemize}
\item from $\NP \cup \coNP$ to $\DP$ by $\NP \cup \coNP \supseteq \NP$, $\NP \cup \coNP \supseteq \coNP$, $\NP \land \coNP = \DP$, and Corollary \ref{cor:closure logic and} following from Köbler, Messner, and Torán \cite{kmt03}.
\item from $\DP$ to $\coDP$ by Corollary \ref{cor:dptocodp}.
\item from $\coBH_k$ to $\coBH_{k+1}$ for $k \geq 2$ by Corollary \ref{cor:closure logic and} following from Köbler, Messner, and Torán \cite{kmt03} and the following inclusions:
\begin{align*}
    \mathrm{coBH}_k \land \mathrm{coBH}_k \supseteq&\ \mathrm{coDP} \land
    \mathrm{coBH}_k = (\mathrm{NP} \lor \mathrm{coNP}) \land \mathrm{coBH}_k\\
     \supseteq&\ \begin{cases}\mathrm{coNP}\land
        \mathrm{coBH}_{k} = \mathrm{coBH}_{k+1} & \text{if } k \text{ is even}\\
    \mathrm{NP} \lor \mathrm{coBH}_{k} = \mathrm{coBH}_{k+1}
    &\text{else}\end{cases}
\end{align*}
\end{itemize}
Next, we derive the equivalence classes for optimal proof systems (see Figure \ref{fig:conclusion}, right, solid arrows). Optimal proof systems translate as follows:
\begin{itemize}
\item from $\coNP$ to $\NP \cup \coNP$ by the fact that $\NP$ has optimal proof systems.
\item from $\NP \cup \coNP$ to $\DP$ by Corollary \ref{cor:closure logic and} following from Köbler, Messner, and Torán \cite{kmt03}.
\item from $\coBH_k$ to $\coBH_{k+1}$ for $k \geq 2$ by the same argument used for p-optimal proof systems.
\end{itemize}
The resulting equivalence classes for (p-)optimal proof systems are different relative to oracles $A$, $B$, $O$ with the following properties (see Figure \ref{fig:conclusion}, left and right, dashed arrows):
\begin{itemize}
\item $\NP^A$ has p-optimal proof systems and $\coNP^A$ has no optimal proof systems \cite{kha22}.
\item $\coNP^B$ has p-optimal proof systems and $\NP^B$ has no p-optimal proof systems \cite{kha22}.
\item $\coNP^O$ has p-optimal proof systems and $\coDP^O$ has no optimal proof systems (Cor.\ \ref{cor:dpnottocodp}).
\end{itemize}
\ifblock{
\begin{figure}[h]
    \begin{centering}
        \includegraphics[width=0.8\textwidth]{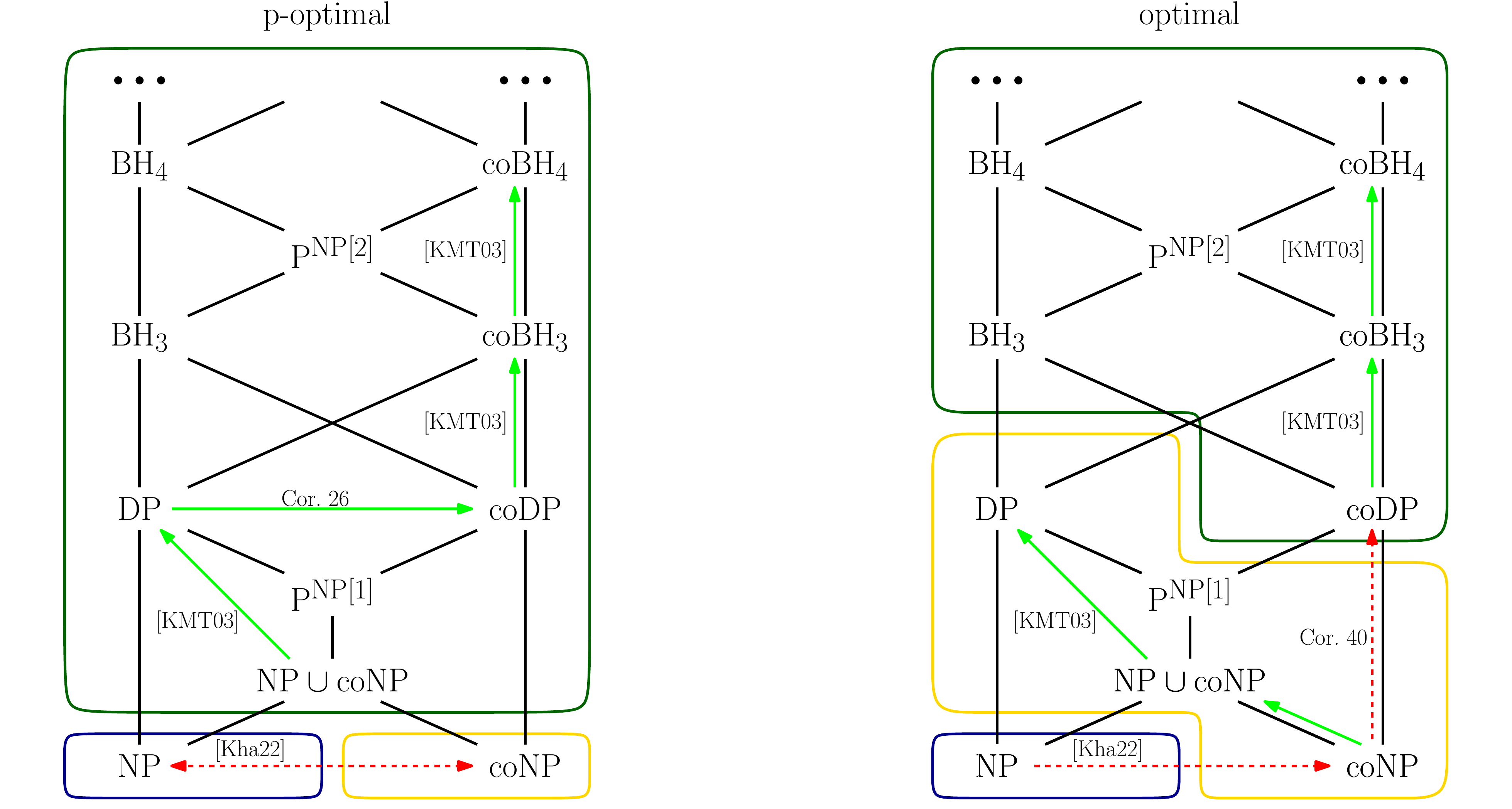}
          \caption{Equivalence classes for p-optimal proof systems (left) and optimal proof systems (right)
          in the Boolean hierarchy over
          $\NP$ and the bounded query hierarchy over $\NP$. Green solid arrows from
          $A$ to $B$ mean that (p-)optimal proof systems for $A$ imply
          (p-)optimal proof systems for $B$. Red dashed arrows from $A$ to $B$ mean
          that there exists an oracle $Q$ relative to which $A^Q$ has (p-)optimal proof
          systems and $B^Q$ has no (p-)optimal proof systems. Note that green solid arrows pointing downwards are omitted, since those are trivial and only the minimum number of required red dashed arrows to separate all equivalence classes are drawn.}
            \label{fig:conclusion}
    \end{centering}
\end{figure}}{
\begin{figure}[h]
    \begin{centering}
        \includegraphics[width=0.8\textwidth]{ipe/conclusion-MFCS.pdf}
          \caption{Equivalence classes for p-optimal proof systems (left) and optimal proof systems (right)
          in the Boolean hierarchy over
          $\NP$ and the bounded query hierarchy over $\NP$. Green solid arrows from
          $A$ to $B$ mean that (p-)optimal proof systems for $A$ imply
          (p-)optimal proof systems for $B$. Red dashed arrows from $A$ to $B$ mean
          that there exists an oracle $Q$ relative to which $A^Q$ has (p-)optimal proof
          systems and $B^Q$ has no (p-)optimal proof systems. Note that green solid arrows pointing downwards are omitted, since those are trivial and only the minimum number of required red dashed arrows to separate all equivalence classes are drawn.}
            \label{fig:conclusion}
    \end{centering}
\end{figure}
}

Oracle $A$ rules out translations from $\NP$ to any other class in Figure \ref{fig:conclusion} for optimal and p-optimal proof systems. Oracle $B$ rules out translations from $\coNP$ to $\NP$ and thus also to $\NP \cup \coNP$ for p-optimal proof systems. Oracle $O$ rules out translations from $\coNP$ to $\coDP$ for optimal proof systems. 

We obtain the following connection to a conjecture studied by Pudlák \cite{pud17}.
\begin{corollary}
The following statements are equivalent:
\begin{itemize}
\item $\BH$ has no p-optimal proof system.
\item $\TAUT$ has no p-optimal proof systems or $\SAT$ has no p-optimal proof systems (i.e., $\mathsf{CON} \lor \mathsf{SAT}$ in Pudlák's notation).
\end{itemize}
\end{corollary}
\begin{proof}
Figure \ref{fig:conclusion} shows that $\NP \cup \coNP$ and $\BH$ are equivalent with respect to p-optimal proof systems. Hence, $\BH$ has no p-optimal proof systems if and only if $\NP \cup \coNP$ has no p-optimal proof systems. The latter holds if and only if $\SAT$ has no p-optimal proof systems or $\TAUT$ has no p-optimal proof systems.
\end{proof}

\bibliography{paper}

\end{document}